\documentclass[twocolumn,amsthm]{autart}

\usepackage{savesym}
\savesymbol{AND}
\savesymbol{FUNCTION
}
\usepackage{lipsum}
\usepackage{amsfonts}
\usepackage{graphicx}
\usepackage{epstopdf}
\usepackage{mathtools}
\usepackage{IEEEtrantools}
\usepackage{bbm}
\usepackage{bm}
\usepackage{enumerate}
\usepackage{amsmath,amssymb}
\usepackage{booktabs}
\usepackage{scalerel,stackengine}
\usepackage[export]{adjustbox}

\usepackage{tikz}
\usepackage{tikzsymbols}
\usetikzlibrary{shapes,arrows,backgrounds,arrows.meta,positioning}
\usepackage[T1]{fontenc}
\usetikzlibrary{arrows}
\usetikzlibrary{calc}
\usepackage{pgfplots}
\usepackage{pgfgantt}
\usepackage{pdflscape}
\pgfplotsset{compat=1.12}
\pgfplotsset{plot coordinates/math parser=false}
\newlength\figureheight
\newlength\figurewidth
\usepackage{dashrule}
\theoremstyle{plain}
\newtheorem{theorem}{Theorem}[section]
\theoremstyle{plain}
\newtheorem{proposition}[theorem]{Proposition}
\theoremstyle{plain}
\theoremstyle{definition}
\newtheorem{definition}[theorem]{Definition}
\theoremstyle{plain}
\newtheorem{lemma}[theorem]{Lemma}
\theoremstyle{plain}

\theoremstyle{plain}
\newtheorem{remark}[theorem]{Remark}
\newtheorem{assumption}[theorem]{Assumption}
\theoremstyle{definition}

\newcommand{\RR}{\mathbb{R}}

\newcommand{\mIntInt}[2]{\mathcal I_{[#1, #2]}}
\newcommand{\mIntGeq}[1]{\mathcal I_{\geq #1}}
\newcommand{\mIndi}{\mathbb I}
\newcommand{\mNorm}[2]{\left\lVert {#2} \right\rVert_{#1}}
\newcommand{\mNormGen}[1]{\left\lVert {#1} \right\rVert}
\newcommand{\mNormSmall}[2]{\lVert {#2} \rVert_{#1}}
\newcommand{\mDefFunction}[3]{#1: #2 \rightarrow #3}

\newcommand{\mSafe}{\mathcal{S}}

\newcommand{\mDef}{\coloneqq}

\newcommand{\mSetValuedConfidence}[1]{\mCS_{#1}}
\newcommand{\mConfidenceSet}[2]{\mathcal C_{#2}(#1)}
\newcommand{\mExpectation}[1]{\mathbb E\left[#1\right]}

\newcommand{\mData}{\mathcal D}

\newcommand{\mOnes}[1]{\mathbbm 1_{#1}}

\newcommand{\mHausdorff}[2]{d_{\mathrm H}(#1, #2)}

\newcommand{\mZpred}{\mu}
\newcommand{\mZpredOpt}{{\mu}^*}
\newcommand{\mZpredCand}{\tilde {\mu}}
\newcommand{\mVpred}{{v}}
\newcommand{\mVpredOpt}{{v}^*}
\newcommand{\mVpredCand}{\tilde {v}}

\newcommand{\mUpred}{u}
\newcommand{\mUpredOpt}{u^*}
\newcommand{\mUpredCand}{\tilde {v}}
\newcommand{\mXpred}{x}
\newcommand{\mXpredOpt}{x^*}

\newcommand{\mCS}{\mathcal E}
\newcommand{\XX}{\mathbb X}
\newcommand{\UU}{\mathbb U}

\newcommand{\NN}{\mathcal {N}}

\newcommand{\ZZ}{\mathbb{Z}}
\newcommand{\LL}{\mathcal{L}}

\newcommand{\mCol}[1]{\mathrm {col}_#1}
\newcommand{\mRow}[1]{\mathrm {row}_#1}

\newcommand*{\END}{\hfill\ensuremath{\lhd}}

\DeclareMathOperator*{\arginf}{arginf}

\tikzstyle{block} = [draw, rectangle,
minimum height=1cm, minimum width=3cm, text width=2.5cm, align=center]
\tikzstyle{sum} = [draw, fill=blue!20, circle, node distance=1cm]
\tikzstyle{input} = [coordinate]
\tikzstyle{output} = [coordinate]
\tikzstyle{pinstyle} = [pin edge={to-,thin,black}]
\makeatletter
\newif\ifmygrid@coordinates
\tikzset{/mygrid/step line/.style={line width=0.80pt,draw=gray!80},
	/mygrid/steplet line/.style={line width=0.25pt,draw=gray!80}}
\pgfkeys{/mygrid/.cd,
	step/.store in=\mygrid@step,
	steplet/.store in=\mygrid@steplet,
	coordinates/.is if=mygrid@coordinates}
\def\mygrid@def@coordinates(#1,#2)(#3,#4){%
	\def\mygrid@xlo{#1}%
	\def\mygrid@xhi{#3}%
	\def\mygrid@ylo{#2}%
	\def\mygrid@yhi{#4}%
}
\newcommand\DrawGrid[3][]{%
	\pgfkeys{/mygrid/.cd,coordinates=true,step=1,steplet=0.2,#1}%
	\draw[/mygrid/steplet line] #2 grid[step=\mygrid@steplet] #3;
	\draw[/mygrid/step line] #2 grid[step=\mygrid@step] #3;
	\mygrid@def@coordinates#2#3%
	\ifmygrid@coordinates%
		\draw[/mygrid/step line]
		\foreach \xpos in {\mygrid@xlo,...,\mygrid@xhi} {%
				(\xpos,\mygrid@ylo) -- ++(0,-3pt)
				node[anchor=north] {$\xpos$}
			}
		\foreach \ypos in {\mygrid@ylo,...,\mygrid@yhi} {%
				(\mygrid@xlo,\ypos) -- ++(-3pt,0)
				node[anchor=east] {$\ypos$}
			};
	\fi%
}
\makeatother
\usepackage{algpseudocode}
\usepackage[authoryear]{natbib}

\begin{document}

\begin{frontmatter}

\title{A predictive safety filter for learning-based control of constrained nonlinear dynamical systems} 

\thanks[footnoteinfo]{This work was supported by the Swiss National Science
Foundation under grant no. PP00P2 157601 / 1. The material in this paper was not
presented at any conference.}

\author{Kim P. Wabersich}\ead{wkim@ethz.ch},
\author{Melanie N. Zeilinger}\ead{mzeilinger@ethz.ch}
\address{Institute for Dynamic Systems and Control, ETH Zurich, Zurich, Switzerland}


\begin{abstract}                          
	The transfer of reinforcement learning (RL) techniques into
	real-world applications is challenged by safety requirements
	in the presence of physical limitations. Most RL methods, in particular
	the most popular algorithms, do	not support explicit
	consideration of state and input constraints. In this
	paper, we address this problem for nonlinear systems
	with continuous state and input spaces by introducing a 
	predictive safety filter, which is able to turn a constrained
	dynamical system into an unconstrained safe system and
	to which any RL algorithm can be applied `out-of-the-box'.
	The predictive safety filter receives
	the proposed control input and decides, based on the current
	system state, if it can be safely applied to the real system,
	or if it has to be modified otherwise. Safety is thereby
	established by a continuously updated safety policy,
	which is based on a model predictive control formulation using
	a data-driven system model and considering state and input
	dependent uncertainties.
\end{abstract}

\end{frontmatter}

\label{sec:problem_statement}

\section{Introduction}
Reinforcement learning (RL) has demonstrated its success in
solving complex and high-dimensional control tasks, see for example
\cite{Levine2016}. These results motivate a more widespread transfer to real-world
applications to enable automated design of high performance controllers
with little need for expert knowledge. In physical systems, such
as mechanical, thermal, biological, or chemical systems, physical
limitations naturally arise as constraints, such as limited torque
in the case of a robot arm or a limited power supply in building control.
In addition to physical constraints, many relevant applications in industry
require satisfaction of safety specifications, preventing, e.g., an autonomous car or
aircraft from crashing, which can typically be formulated in terms of
constraints on the system state.
The simultaneous satisfaction of safety constraints
under physical limitations during RL constitutes one of the main open problems in
AI safety as discussed e.g. in~\citet[Section 3]{amodei2016concrete}.

Significant progress in the safe operation of constrained systems has been made
through model predictive control techniques, which provide rigorous constraint
satisfaction, see, e.g., \cite{Mayne2014}. While model-based RL techniques such as
\cite{Kamthe2017} are conceptually closely related to model predictive control,
so far relatively few methods have considered safety guarantees. Learning-based
model predictive control aims to combine the benefits of both fields,
see for example~\cite{Hewing2020} for an overview. In addition to the fact
that designing such algorithms with rigorous safety guarantees is rather
challenging, often conservative, and requires a considerable amount of
expert knowledge, the approach is inherently restricted to a model-based
control policy. More precisely, at each time step, a finite-horizon optimal
control problem is solved in a receding horizon fashion in order to approximate
a potentially infinite horizon optimal control policy.

\emph{Concept:}
We propose a model predictive control (MPC) variant
as a predictive safety filter (PSF), that can turn highly nonlinear
and safety-critical dynamical systems into inherently \emph{safe systems},
and to which any RL algorithm without safety certificates can be
applied `out-of-the-box', see also Figure~\ref{fig:concept}.
Compared with a standard use of MPC, the PSF verifies if the input
proposed by the RL algorithm is safe, otherwise it is entitled to
modify the input as little as necessary to maintain safe operation
at all future times. This means that the PSF
only needs to keep the system safe instead of controlling it
well with respect to a certain objective
(e.g. comfort or economic criteria). The problem of finding a safety filter is therefore in
general less complex than finding a desired optimal policy with respect to some
objective and subject to constraints, motivating the combination 
of a predictive safety filter with an RL algorithm to safely optimize performance.

Differently to recently proposed related concepts presented
in~\cite{Gillula2011,Wabersich2018,Ames2019},
we use the notion of a \emph{safe system} in Figure~\ref{fig:concept}, as similarly
introduced in \cite{Wieland2007} within the context of safety barrier functions.
The concept emphasizes the possibility that \emph{any} RL algorithm that
would have been used to control the original system
can be applied to the \emph{safe system} instead, yielding a certified safe RL application.
The predictive safety filter provides \emph{safety} at a desired level of probability,
\emph{modularity} in terms of the employed RL controller, and \emph{minimal intervention}
by filtering RL input signals only if we cannot guarantee safety at
the specified probability level, similar to \cite{Fisac2019}.
\begin{figure}
	\centering
	\begin{tikzpicture}[scale=0.8]
		\input{fig/controlLoopSafe.tex}
	\end{tikzpicture}
	\caption{Concept of predictive safety filter: Based on the current state
		$x(k)$, a learning-based algorithm provides a control input
		$u_\LL(k)=\pi_\LL(k,x(k))\in\RR^m$, which is processed by the safety filter
		$u(k) = \pi_\mSafe(k,x(k),u_\LL(k))$ and applied to the real system.}
	\label{fig:concept}
\end{figure}

\emph{Contributions:}
Based on a probabilistic model of the system dynamics, which is inferred from data,
this paper presents a predictive safety filter that builds on concepts from MPC for
constrained nonlinear systems, and thereby generalizes the safety
certification method for \textit{linear} systems proposed by~\cite{Wabersich2018}.
Safety of an RL input is thereby enforced in real-time
by searching for a safe backup trajectory for the next time step towards a known
set of safe states. If necessary, to ensure safety for all times in the future, the
search process for a backup trajectory is allowed to modify (filter)
the RL input.
As MPC typically outperforms non-optimization based techniques, e.g. based on
control Lyapunov functions or sliding-mode controllers
by solving an approximate optimal control problem on-line, the proposed PSF
provides similar advantages compared with methods utilizing, e.g., control
barrier functions for safety~\citep{Ames2019}.
More precisely, the PSF formulation provides an implicit representation of the
set of safe state and input pairs, approximating the largest set of
admissible states and inputs using finite-horizon MPC techniques.
At the same time, the implicit safe set representation enables
favorable scalability properties compared to, e.g., Hamilton-Jacobi-Bellman
safety frameworks~\citep{Fisac2019} by avoiding offline computations
that scale exponentially with the number of state dimensions.
Clearly, these advantages come at the price of solving an optimization
problem online, for which highly efficient tools are, however,
available~\citep{domahidi2012efficientInteriorPoint}.

The application to nonlinear and probabilistic
system descriptions, obtained, e.g., through machine learning techniques, is
enabled via a nonlinear predictive safety filter formulation
that is robust in probability and supports state and input dependent
uncertainty information. Robustness with respect to non-uniform
model uncertainties is enabled by restricting predicted backup trajectories to
confident subsets of the state and input space. The corresponding
online optimization problem has similar computational complexity to
nominal MPC problems, while being less conservative compared to other robust MPC
approaches~\citep[Section 3.1]{Hewing2020} that are often based on a
uniform uncertainty bound.
The proposed formulation leads to a theoretical analysis that rigorously relates parameters of
the predictive safety filter and accuracy of its system model to safety in probability.
Depending on the desired constraint satisfaction probability,
this enables safe exploration beyond available data.

We illustrate the approach using a simulated pendulum
swing-up task, in which only little initial data around the
stable downward position is available and overshoots of the upward
position are prohibited, imposing challenging safety constraints
on the system. Scalability and practical implementation are demonstrated
in a quadrotor simulation example, considering the task of
learning to quickly approach a landing position for a full-scale model with 12
states and 4 inputs.

\emph{Discussion:}
While the focus of this paper is the certification of RL algorithms, the
concept can also be used together with, e.g., human inputs. For example,
in the case of autonomous driving, the safety filter could be used
to ensure safety of either an RL-based controller or a human driver, and
can be viewed as a driver assistance system that is able to overrule the
student driver (or RL algorithm), if necessary for safety.

\section{Related Work}

\emph{Safe model-free reinforcement learning:}
There is a growing awareness of safety questions in the domain of
artificial intelligence \citep{amodei2016concrete}, and several
safe reinforcement learning techniques have been proposed, see e.g.
\cite{Garcia2015} for an overview. \cite{achiam2017constrained}, e.g., provide
safety in expectation based on a trust-region approach with respect
to the policy gradient.

Most notions of safety considered in this line of research, e.g.
one-step constraint satisfaction in expectation, tend to be less strict
compared with the probabilistic safety requirements at all time steps in the future as
considered in this paper. More importantly, since most techniques are policy-based,
safety is coupled to a specific policy and therefore potentially also to a specific task,
limiting generalization of the safety certificates.

\emph{Learning-based model predictive control:}
Originating from concepts in robust model predictive control (MPC),
extensions of MPC schemes to safe learning-based methods 
have been proposed, see e.g. \cite{Hewing2020} for a review.
In addition, various results have investigated combinations of
MPC with learning-based online model identification techniques
\citep{Ostafew2016,limon2017learning,Koller2018,soloperto2018learning},
also in an adaptive manner \citep{Tanaskovic2013}.
In the context of robotics, similar concepts exist, which are often referred to
as funneling, see e.g. \cite{Majumdar2017} and references therein, as well as so called
LQR-trees \citep{Tedrake2010}.

While some of these techniques have been demonstrated to work well in practice
\citep{bouffard2012learning,Ostafew2016,Hewing2018b},
they typically either lack rigorous theoretical safety guarantees, tend to be overly conservative by
relying on Lipschitz-based estimates in the prediction of the uncertain system evolution, or
are restricted to a very specific system class.

\emph{Model-based policy certification through barrier functions and safety frameworks:}
Conceptually, the idea of safety architectures,
as illustrated in Figure~\ref{fig:concept}, was originally
proposed by~\cite{Seto1998}, where a principled switching
between a safety controller and basic/experimental controllers enables safe
controller tuning online. Based on this concept, control
theoretic frameworks have been developed~\citep{Prajna2004},
related to control Lyapunov functions~\citep{Wieland2007},
and have also become known as control barrier
functions, see~\cite{Ames2019} for an overview.
Recent developments also consider the combination with learning
tasks through data-driven models~\citep{Ohnishi2019}. 
While such frameworks inherit strong theoretical results from
control Lyapunov function theory, they require the explicit availability of a
control barrier function, which is difficult to compute in general.
In particular, to the best of the authors' knowledge, only
approximate approaches exist using ellipsoidal or sum-of-squares
computations~\citep{Wang2018a} to design a safety barrier
function with respect to given, e.g. polytopic, state and input constraints. 

In the case of partially known system dynamics, the concept of
a control barrier function can be combined with Bayesian model
estimates from data that validate the resulting closed-loop
system~\citep{berkenkamp2016safe}.
The techniques share similar limitations with
safe model-free RL methods, namely that they are tailored to a specific task.
A task-independent learning-based safety framework has been introduced in \cite{Gillula2011},
which generalizes the concept of explicitly knowing a barrier function.
It consists of a model-based safe set of system states
and computes a corresponding safe control policy, which is entitled to override a potentially
unsafe RL algorithm to ensure invariance with respect to the safe set of system states
i.e. containment within the safe set at all times.
This concept was further developed in several papers, providing principled methods to
compute the safe set as well as the corresponding safe policy
\citep{Fisac2019,wabersich2017scalableSafety},
which build the foundation of the safety filter presented in this paper.

The aforementioned techniques related to both safety barrier functions and safety frameworks
either suffer from limited scalability to higher dimensional and complex systems, or 
only provide principled design computations for  
specific constrained model classes, such as linear or polynomial models.
While also building on the same high-level concept, this paper addresses these limitations by
1) considering a stochastic nonlinear system model belief, which is well-suited to
learning-based control of highly nonlinear and unstable system dynamics, and 2) a unified
MPC-inspired formulation for the safety policy
(predictive safety filter), which avoids the explicit computation of a
safety barrier function or a safe set.

In particular, compared
to similar MPC-inspired safety mechanisms such as
\citet{Wabersich2018,Mannucci2018,li2019robust},
we consider nonlinear system models with stochastic
parameter uncertainties and provide a predictive
safety filter formulation that is capable of leveraging the resulting state
and input dependent uncertainty estimates
to reduce conservatism.

\section{Problem Statement}
\emph{Notation:}
The set of integers in the interval $[a,b]\subset\RR$ is
$\mIntInt{a}{b}$, and the set of integers in the interval
$[a,\infty)\subset\RR$ is $\mIntGeq{a}$.
The $i$-th row and $i$-th column of
a matrix $A\in\RR^{n\times m}$ is denoted by $\mRow{i}(A)$
and $\mCol{i}(A)$. By $\mOnes{n}$ we denote the vector of ones
with length $n$. \END

Consider deterministic discrete-time systems of the form
\begin{align}
	x(k+1) = f(x(k), u(k); \theta_{\mathcal R}), ~ \forall k\in\mIntGeq{0},
	\label{eq:general_nonlinear_system}
\end{align}
with dynamics $\mDefFunction{f}{\XX \times \UU}{\RR^n}$ parametrized
by $\theta_{\mathcal R}\in \RR^p$ and stochastic initial condition
$x(0) = x_{\mathrm{init}}\in\XX$ with known distribution $p(x_{\mathrm{init}})$.
The system is subject to polyhedral state and input constraints
$\XX \mDef \{x \in \RR^n \vert A_x x \leq \mOnes{n_x} \}$ and
$\UU \mDef \{u \in \RR^m \vert A_u u \leq \mOnes{n_u} \}$, originating
from physical limitations and safety requirements. We consider the case
of unknown `real' parameters $\theta_{\mathcal R}$, but assume the availability
of a distribution
\begin{align}
	\theta \sim p(\theta) \text{ with mean }\mExpectation{\theta},
	\label{eq:dynamics_model}
\end{align}
which can be estimated from data. The overall objective is to \emph{safely} find a policy
$\mDefFunction{\pi_\LL}{\mIntGeq{0}\times\XX}{\UU}$ that either minimizes an episodic,
finite-time or infinite horizon objective
\begin{align}\label{eq:general_objective}
	J_{\bar N}(x(k)) =
		\mExpectation{\sum_{i=k}^{\bar N} \ell\big(x(i), \pi_\LL(i,x(i))\big)}
\end{align}
with $\bar N \in \mathbb N$, and stochastic stage cost
$\ell(x, u) \mDef \bar \ell(x,u) + w_\ell$ consisting of a deterministic
part $\mDefFunction{\bar \ell}{\XX \times \UU}{\RR}$ and zero mean i.i.d.
stochastic noise $w_\ell$. In order to prescribe a desired level of caution
and desired conservatism in exploration, we consider \emph{safety} in terms
of constraint satisfaction at a desired probability level $p_\mSafe>0$ as
\begin{align}\label{eq:chance_constraints}
	\Pr\left(\forall k \in \mIntInt{0}{\bar N}: x(k)\in\XX, u(k)\in\UU\right) & \geq p_\mSafe.
\end{align}
This paper addresses the problem of implementing a safety filter as shown in
Figure~\ref{fig:concept}, which ensures closed-loop safety according
to~\eqref{eq:chance_constraints}.
The filter enables application of any RL algorithm 
to the virtual input of the safe system, i.e. $u_\LL(k) = \pi_\LL(k,x(k))$,
with the goal of minimizing the objective, while
ensuring safety by selecting the input to the real system as
$u(k) = \pi_\mSafe(k,x(k),u_\LL(k))$.
In other words, the approach turns a safety-critical task into
an unconstrained task with respect to the safe system dynamics
$f_\mSafe(k,x(k),u_\LL(k)) \mDef f(x(k),\pi_\mSafe(k,x(k),u_\LL(k)))$
such that any RL algorithm can be safely applied, see Figure~\ref{fig:concept}.
To further specify the desired properties of $\pi_\mSafe$,
consider the following definition of a safety certified learning-based control input.
\begin{definition}\label{def:safety_certified_input}
	An input $u_\LL(\bar k)$ is \emph{certified as safe} for system
	\eqref{eq:general_nonlinear_system} at time step $\bar k$ and
	state $x(\bar k)$ with respect to a \emph{safety filter}
	$\mDefFunction{\pi_\mSafe}{\mIntGeq{0}\times\XX\times\UU}{\UU}$,
	if $\pi_\mSafe(\bar k,x(\bar k),u_\LL(\bar k))=u_\LL(\bar k)$ and
	application of $u(k)=\pi_\mSafe(k,x(k),u_\LL(k))$ for $k\geq\bar k$
	implies safety for all times according to \eqref{eq:chance_constraints}.
\end{definition}
Following this definition, the goal is to provide a safety filter
$\pi_\mSafe$ that restricts learning as little as possible by certifying
a large set of learning inputs $u_\LL(k)$ for a given state $x(k)$.
If the learning input cannot be certified as safe,
the safety filter provides an alternative safe input, i.e. 
$u(k)=\pi_\mSafe(k,x(k),u_\LL(k))\neq u_\LL(k)$, where 
the filter aims at the smallest possible modification by, e.g.,
minimizing $\mNorm{2}{\pi_\mSafe(k,x(k),u_\LL(k)) - u_\LL(k)}$.
The following section introduces the mechanisms of the proposed predictive safety filter,
which builds on a predictive constrained control formulation, planning safe
trajectories based on a probabilistic model belief to ensure safe system operation at all times
according to~\eqref{eq:chance_constraints}.


\begin{figure*}[p]
	\begin{minipage}[b]{.4\textwidth}
		\begin{center}{\large Nominal Predictive Safety Filter}\end{center}
		\vspace{0.5cm}
		\textbf{Nominal online problem:}
		\begin{subequations}\label{eq:MPSF_nominal_opt}
			\begin{align} \label{eq:MPSF_nominal_opt_cost}
				\min_{\substack{\{\mUpred_{i|k}\}}}
				~           & \mNormGen{u_\LL - \mUpred_{0|k}}                                                                                          \\
				\nonumber
				\text{s.t.} & ~ \forall i\in\mIntInt{0}{N-1}:                                                                                           \\
				            & \mXpred_{i+1|k} = f(\mXpred_{i|k}, \mUpred_{i|k};\bar \theta),~ \label{eq:MPSF_nominal_opt_dynamic_constraints} \\
				            & \mXpred_{i|k} \in \XX,~  				   		\label{eq:MPSF_nominal_opt_state_constraints}                                                 \\
				            & \mUpred_{i|k} \in \UU,~ 					   		\label{eq:MPSF_nominal_opt_input_constraints}                                                  \\
				            & (\mXpred_{i|k}, \mUpred_{i|k})\in\ZZ_c,~	\label{eq:MPSF_nominal_opt_uncertainty_constraints}                          \\
				            & \mXpred_{N|k} \in \mSafe^t,	                \label{eq:MPSF_nominal_opt_safety_constraint}                               \\
				            & \mXpred_{0|k} = x(k).										\label{eq:MPSF_nominal_opt_initial_constraint}
			\end{align}
		\end{subequations}

		\vspace{0.15cm}
		\textbf{Algorithm 1 (Nominal PSF):}

		\vspace{0.15cm}
		\input{algorithms/Nominal_MPSF.alg.tex}

		\vspace{0.25cm}
		\textbf{Illustration of nominal PSF:}\\~\\
		\begin{tikzpicture}[scale = 1.05]
			\vspace*{1cm}
			\input{fig/mpsf_scheme_nom.tex}
		\end{tikzpicture}
	\end{minipage}
	\unskip\hfill\vrule\hfill
	\begin{minipage}[b]{.4\textwidth}
		\begin{center}{\large Predictive Safety Filter}\end{center}
		\vspace{0.5cm}
		\textbf{Online problem:}
		\begin{subequations}\label{eq:MPSF_opt}
			\begin{align} \label{eq:MPSF_opt_cost}
				\min_{\substack{\{\mVpred_{i|k}\}}}
				~           & \mNormGen{u_\LL - \mVpred_{0|k}}                                                                                             \\
				\nonumber
				\text{s.t.} & ~ \forall i\in\mIntInt{0}{N-1}:                                                                                          \\
				            & \mZpred_{i+1|k} = f(\mZpred_{i|k}, \mVpred_{i|k}; \bar \theta),~ \label{eq:MPSF_opt_dynamic_constraints}                    \\
				            & \mZpred_{i|k} \in \bar \XX_i,~  				   		\label{eq:MPSF_opt_state_constraints}                                                     \\
				            & \mVpred_{i|k} \in \bar \UU_i,~ 					   		\label{eq:MPSF_opt_input_constraints}                                                      \\
				            & \mSetValuedConfidence{p_\mSafe}(\mZpred_{i|k}, \mVpred_{i|k}) \subseteq \bar \mCS_{i}^\gamma,~	\label{eq:MPSF_opt_uncertainty_constraints} \\
				            & \mZpred_{N|k} \in \bar \mSafe^t_N,	                \label{eq:MPSF_opt_safety_constraint}                                   \\
				            & \mZpred_{0|k} = x(k).										\label{eq:MPSF_opt_initial_constraint}
			\end{align}
		\end{subequations}

		\vspace{0.15cm}
		\textbf{Algorithm 2 (PSF):}
		
		\vspace{0.15cm}
		\input{algorithms/MPSF.alg.tex}

		\vspace{0.25cm}
		\textbf{Illustration of PSF under uncertainty:}\\~\\
		\begin{tikzpicture}[scale = 1.05]
			\input{fig/mpsf_scheme_uncertain.tex}
		\end{tikzpicture}

	\end{minipage}
	\caption{\small The basic idea of the predictive safety filter explained using
		a nominal, simplified version in the left column and the final method on the right.
		The illustrations show the system state at time $k$ with
		safe backup plan for a shorter horizon obtained from the solution at time $k-1$, 
		depicted in brown, and areas with poor model quality in red.
		An arbitrary learning input $u_\LL$ is certified if a feasible solution towards the
		terminal safe set $\mSafe^t$ can be found, as shown in green. If this new backup
		solution cannot be found and the planning problem
		\eqref{eq:MPSF_nominal_opt}/\eqref{eq:MPSF_opt} is infeasible, the system can be driven to
		the safe set $\mSafe^t$ along the brown previously computed trajectory.
		\textbf{Left (NPSF):} By assuming perfect system knowledge, the computed
		backup plans correspond exactly to the true state dynamics and
		constraints are guaranteed to be satisfied using the nominal backup trajectory.
		\textbf{Right (PSF):} Backup plans are computed w.r.t. the
		nominal expected state $\mZpred$. The true state trajectory lies within a growing tube
		around the nominal state with probability $p_\mSafe$, which needs to be considered
		using tightened constraints according to~\eqref{eq:constraint_tightening}.
	}
	\label{fig:model_predictive_safety}
\end{figure*}

\section{Predictive safety filter}

We first develop an intuitive understanding of the predictive
safety filter by considering a simplified setting and assuming perfect model knowledge
in Section~\ref{subsec:nominal_psf}, which is then extended in
Section~\ref{subsec:psf} to an uncertain model \eqref{eq:general_nonlinear_system},
\eqref{eq:dynamics_model} inferred from data, for which rigorous proofs are provided.
As it will be shown, the presented method establishes safety by relying
on controllability of \eqref{eq:general_nonlinear_system} along
system trajectories, in combination with an efficient mechanism enforcing
the system to carefully enter uncertain areas within the state and
input space.

\subsection{Nominal (simplified) predictive safety filter}
\label{subsec:nominal_psf}

Consider the simplified situation where the
real system dynamics \eqref{eq:general_nonlinear_system} are
perfectly known for some subset of the state and input space,
as specified in the following.

\begin{assumption}\label{ass:set_of_perfect_system_knowledge}
	There exists a set $\ZZ_c \subseteq \XX\times\UU$, such that for
	all $(x, u)\in\ZZ_c$ and some $\bar \theta \in \RR^p$ it holds that
	$f(x, u; \bar \theta) = f(x,u;\theta_{\mathcal R})$.
\end{assumption}

Similarly to \cite{Wabersich2018}, we propose a predictive
safety filter that is not pre-computed, but defined via an optimization problem and
computed on-the-fly. The main working mechanism is the construction
of safe backup plans that, if applied, would keep the system provably
safe in the future, see Figure~\ref{fig:model_predictive_safety} (left) for an illustration.
The backup plans are defined via \eqref{eq:MPSF_nominal_opt}, where
$\{\mXpred_{i|k,N}\}$ denote the planned states computed at the
current time step $k$ and predicted $i$ time steps into the future
with planning horizon $N$ using the corresponding input sequence $\{\mUpred_{i|k,N}\}$.
One of the key challenges in computing the backup plans is to deal
with the fact that a good model is not known in unexplored regions
of the state-space, i.e. $\XX \setminus \ZZ_c$, shown as red
(unconfident model) sets in Figure~\ref{fig:model_predictive_safety}.
In the nominal setting, we simply address this problem by enforcing the system to
strictly stay within the confident model subset $\ZZ_c$ via \eqref{eq:MPSF_nominal_opt_uncertainty_constraints}.
One of the main problems addressed in the next section will be to relax this constraint
to enable cautious exploration of such unconfident subsets. 
The purpose of the remaining constraints in \eqref{eq:MPSF_nominal_opt} is to
construct backup plans that lead the system within state and input
constraints $\XX$ and $\UU$ \eqref{eq:MPSF_nominal_opt_state_constraints},
\eqref{eq:MPSF_nominal_opt_input_constraints} into a safe terminal set $\mSafe^t$
in $N$ steps \eqref{eq:MPSF_nominal_opt_safety_constraint}.

The objective of constructing the backup plans in \eqref{eq:MPSF_nominal_opt} is
to minimize the deviation between the first element of the input sequence
$\mUpred_{0|k,N}$ and the input $u_\LL(k)$ requested by the RL algorithm,
such that $\mUpred_{0|k,N}=u_\LL(k)$ if $u_\LL(k)$ is safe. Conceptually, this
mechanism is similar to QP-based barrier function methods~\citep{Ames2017}, where
the input is adjusted to remain inside an explicitly known invariant set, with
the key difference that here the safe set is implicitly defined. The resulting
nominal predictive safety filter is then given by
$\pi_{\NN\mSafe}(k,x(k),u_\LL(k))=\mUpredOpt_{0|k,N}$, with $\mUpredOpt_{0|k,N}$
being the optimal first control input obtained from \eqref{eq:MPSF_nominal_opt} based 
on a prediction horizon of length $N$.
To ensure constraint satisfaction beyond the planning horizon,
\eqref{eq:MPSF_nominal_opt} utilizes a mechanism common
in predictive control (see e.g. \cite{chen1998quasiInfiniteHorizonMPC}),
by requiring the last state of the sequence $\{\mXpred_{N|k,N}\}$ to lie
in a safe terminal set of system states $\mSafe^t$, for which a
locally valid safety filter $\pi_\mSafe^t$ is known.
\begin{assumption}\label{ass:terminal_safe_set}
	There exists a terminal safe set
	$\mSafe^t\mDef\{x\in\RR^n | a_\mSafe(x)\leq \mOnes{n_\mSafe} \} \subseteq \XX$,
	with $a_\mSafe$ Lipschitz continuous with Lipschitz constant $L_\mSafe$,
	and a corresponding terminal safety filter
	$\mDefFunction{\pi_\mSafe^t}{\mIntGeq{0}\times \XX \times \UU}{\UU}$,
	such that if $x(\bar k)\in\mSafe^t$, then application of
	$u(k)=\pi_\mSafe^t(k,x(k),u_\LL(k))$ implies that
	$x(k)\in\XX$ and $u(k)\in\UU$ for all $k > \bar k$.
\end{assumption}

A terminal safe set $\mSafe^t$ and the corresponding controller $\pi_\mSafe^t$ can be chosen, e.g.,
as a classical terminal set for nonlinear (robust) MPC \citep{chen1998quasiInfiniteHorizonMPC}, 
regions around stable steady-states of system~\eqref{eq:general_nonlinear_system}, or using
expert system knowledge as is demonstrated in Section~\ref{sec:numerical_example}.

Based on problem~\eqref{eq:MPSF_nominal_opt}, the predictive safety filter
$\pi_\mSafe$ is defined by Algorithm~1 (Figure~\ref{fig:model_predictive_safety}, left).
At every time step, we attempt to solve optimization problem~\eqref{eq:MPSF_nominal_opt}.
If problem \eqref{eq:MPSF_nominal_opt} is feasible at time $k$, safety, i.e.,
$x(k)\in\XX$, $u(k)\in\UU$, directly follows from \eqref{eq:MPSF_nominal_opt_state_constraints},
\eqref{eq:MPSF_nominal_opt_input_constraints}.
Due to the generality of the terminal safe set, however, problem \eqref{eq:MPSF_nominal_opt}
may become infeasible for some state $x(k)$, even after being feasible at the previous time step $x(k-1)$. 
Algorithm~1 implements a shrinking horizon mechanism
similar to~\cite{Thomas1994,Grune2014} to also provide a feasible
safe trajectory and input sequence towards the terminal safe set
for this case, as detailed in the following:

Assume that \eqref{eq:MPSF_nominal_opt} was feasible at time
$k-1$ with corresponding optimal input sequence
$\{\mUpredOpt_{i|k-1,N}\}$. Application of $u(k-1) = \mUpredOpt_{0|k-1,N}$
results in a safe state $x(k)$ as depicted in Figure~\ref{fig:model_predictive_safety} (left),
because $(x(k-1), u(k-1))\in\ZZ_c$ by \eqref{eq:MPSF_nominal_opt_uncertainty_constraints}
and therefore $x(k) = f(x(k-1), u(k-1);\bar \theta)\in\XX$ by \eqref{eq:MPSF_nominal_opt_state_constraints}.
At the next time step $k$, if \eqref{eq:MPSF_nominal_opt} is not feasible, we can still solve \eqref{eq:MPSF_nominal_opt} with
a reduced planning horizon $N-1$. This can be easily verified by noting that $\mUpred_{i|k} = \mUpredOpt_{i+1|k-1}$ for
$i\in\mIntInt{0}{N-2}$, i.e. the tail of the previously computed feasible trajectory from time step $k-1$,
is a feasible solution as depicted by the brown trajectory in
Figure~\ref{fig:model_predictive_safety} (left). Feasibility of \eqref{eq:MPSF_nominal_opt}
for a reduced horizon again directly provides $x(k)\in\XX$, $u(k)\in\UU$.

The same holds true in the case that $j<N$ steps were consecutively
infeasible for planning horizon $N$, i.e. \eqref{eq:MPSF_nominal_opt} will then be
feasible with horizon $N-j$ until we reach the safe terminal set.
This shortening of the horizon is implemented in lines 6-7 of Algorithm 1.
If the horizon length reaches $0$, the state is in the terminal set and
$\pi_\mSafe^t$ can be applied to ensure $x(k)\in\XX$, $u(k)\in\UU$ (line 9).
Note again that if \eqref{eq:MPSF_nominal_opt} is feasible at time $k$ (line 3-4),
$\mUpredOpt_{0|k,N}$ can be applied, which ideally results in $u_\LL(k)$
(i.e. objective \eqref{eq:MPSF_nominal_opt_cost} is zero) as shown in
Figure~\ref{fig:model_predictive_safety} (left) together with the
optimal backup plan in green. Algorithm 1 therefore ensures
constraint satisfaction at all time steps,
realizing a predictive safety filter in a receding horizon fashion
with varying prediction length.
The next section will extend the previously introduced basic concept of
the predictive safety filter to consider a data-driven approximate system belief,
represented by \eqref{eq:general_nonlinear_system}, \eqref{eq:dynamics_model},
subject to probabilistic constraint satisfaction \eqref{eq:chance_constraints}.

\subsection{Predictive safety filter}\label{subsec:psf}
A key goal of the safety filter is to support exploration beyond available data
via the learning policy $\pi_\LL(k)$, in which case
Assumption~\ref{ass:set_of_perfect_system_knowledge} does not necessarily hold.
While fast approximate computation of the backup trajectories can still be performed
online using the mean estimate $\bar \theta$ of the parameter $\theta_{\mathcal R}$,
we need to safely handle the resulting non-vanishing model error
\begin{align}\label{eq:state_input_dependent_model_error}
	e(k,\theta_{\mathcal R}) \mDef f(x(k),u(k);\theta_{\mathcal R}) - f(x(k), u(k); \bar \theta).
\end{align}
In the following, we first treat uncertainty via a uniform error
bound to introduce the safety filter for uncertain systems,
which is then extended to consider a less conservative bound and
impose it as a constraint in the filter planning problem,
in order to reduce conservatism.

\emph{Uniformly bounded model error:}
Assume that the model error
with respect to the point estimate $\bar \theta$ can be bounded as
\begin{align}\label{eq:uniform_error_set}
 &\Pr(e(k,\theta_{\mathcal R})\in\mathcal E \text{ for all~} k\in\mIntGeq{0})\geq p_\mSafe
\end{align}
with $\mCS \mDef \{e\in\RR^n | a_\mCS(e) \leq \mOnes{n_\mCS}\}$ and
$\mDefFunction{a_\mCS}{\RR^n}{\RR^{n_\mCS}}$ Lipschitz continuous and
linearly bounded from below, i.e. constants $L_{a_{\mCS}}>0$ and $c_{\mCS}>0$
exist such that $c_\mCS \mNormSmall{2}{e}\leq \mNormSmall{2}{a_\mCS(e)}
\leq L_{a_{\mCS}}\mNormSmall{2}{e}$, which implies compactness of $\mCS$.
In this case, the filter can still compute backup plans using the point estimate
$\bar \theta$, however, in contrast to the nominal case in Section~\ref{subsec:nominal_psf},
the constraints in \eqref{eq:MPSF_nominal_opt} are modified such that prediction
errors induced by \eqref{eq:state_input_dependent_model_error} are compensated to
ensure constraint satisfaction.

We denote the nominal (expected) system states as $\{\mZpred_{i|k}\}$,
corresponding to the nominal input sequence $\{\mVpred_{i|k}\}$ according to
$\mZpred_{i+1|k} = f(\mZpred_{i|k}, \mVpred_{i|k}; \bar \theta)$.
Due to the model error \eqref{eq:state_input_dependent_model_error},
we need to address the fact that potentially $x(k+1) \notin \XX$,
i.e. $A_x x(k+1) \nleq \mOnes{n_x}$, when applying the nominal input
$\mVpred_{0|k}$, even though the corresponding nominal predicted state satisfies
$\mZpredOpt_{1|k}\in\XX$.
A common strategy for achieving robustness in predictive control 
is to tighten the constraints by leveraging controllability along any possible
predicted state sequence $\{\mZpred_{i|k}\}$~\citep{Mayne2014}. Intuitively
speaking, controllability enables efficient compensation of deviations
$x(i) - \mZpred_{i|k}$ via feedback control.
More precisely, the possible deviations can be bounded by a decay constant,
expressed by a parameter $\rho$, at which a controller can compensate
disturbances of a certain magnitude, defined proportionally to a parameter
$\epsilon$. Using these two measures, deviations from the planned nominal trajectory
can be compensated via an iterative tightening of the constraints. This allows
a flexible response to upcoming disturbances at the desired
probability level $p_\mSafe$ during consecutive time steps
via replanning, thereby enabling overall constraint satisfaction.
Following \cite{Koehler2018b}, we tighten the constraints
\eqref{eq:MPSF_nominal_opt_state_constraints}, \eqref{eq:MPSF_nominal_opt_input_constraints},
and \eqref{eq:MPSF_nominal_opt_safety_constraint} in the computation of the backup plans as
\begin{subequations}\label{eq:constraint_tightening}
\begin{align}
	\bar\XX_i 		&\mDef \{ x \in \RR^n | A_x x \leq (1-\epsilon_i)\mOnes{n_x}\},         \\
	\bar \UU_i      & \mDef\{ u \in \RR^m | A_u u \leq (1-\epsilon_i)\mOnes{n_u}\},              \\
	\bar \mSafe_N^f & \mDef \{ x \in \RR^n | a_\mSafe(x) \leq (1-\epsilon_N)\mOnes{n_\mSafe}\}, 
\end{align}
\end{subequations}
implementing a trade-off between compensation and magnitude of disturbances
via the converging recursion
\begin{align}
	\begin{rcases*}
		\epsilon_0 \mDef 0\\
		\epsilon_{i+1} \mDef \epsilon_i + \sqrt{\rho}^i \epsilon
	\end{rcases*}
	\Rightarrow \epsilon_i = \epsilon\frac{1-\sqrt{\rho}^i}{1-\sqrt{\rho}},
	\label{eq:constraint_tightening_sequence}
\end{align}
with design parameter $\epsilon>0$ and parameter $\rho\in(0,1)$ that depends on
system \eqref{eq:general_nonlinear_system} as follows.
\begin{assumption}
	\label{ass:local_incremental_stabilizability}
	There exists a control policy $\mDefFunction{\pi}{\XX\times\XX\times\UU}{\RR^m}$,
	a function $\mDefFunction{V}{\XX\times\XX\times\UU}{\RR_{\geq 0}}$, which is
	continuous in its first argument and satisfies $V(x,x,v)=0$ for all
	$x\in\XX$, $v\in\UU$, and parameters $c_l,c_u,\delta,\pi_{\mathrm{max}}\in\RR_{\geq 0}$,
	$\rho\in(0,1)$, such that for a given $\bar \theta\in\RR^q$
	the following properties hold for all $x,\mu\in\XX$, $v,v^+\in\UU$:
	\begin{align*}
		c_l\mNorm{2}{x-\mu}^2 \leq V(x,\mu,v) \leq c_u \mNorm{2}{x-\mu}^2
	\intertext{and if in addition $V(x,\mu,v) \leq \delta$ then}
		\mNorm{2}{\pi(x,\mu,v)-v} \leq \pi_{\mathrm{max}}\mNorm{2}{x-\mu}           
		\\
		V\left(f(x, \pi(x,\mu,v);\bar \theta),f(\mu,v;\bar \theta),v^+ \right) \leq \rho V(x,\mu,v).
	\end{align*}
\end{assumption}
Informally, Assumption~\ref{ass:local_incremental_stabilizability} defines
how well the uncertain system can be controlled in a neighborhood of
predicted nominal backup plans $\{\mZpredOpt_{i|k}\}$. Intuitively speaking,
considering the task of tracking a reference trajectory as an optimal control problem
with value function $V$ (using for example a linear quadratic regulator in
the linear dynamics setting), parameter $\rho$ defines `how fast' a reference
can be reached, measured in terms of the contraction rate of the
optimal tracking cost $V$.
Interestingly, this translates into a system-theoretic requirement on
system~\eqref{eq:general_nonlinear_system}, or more precisely to
local incremental stabilizability, which can be formally verified
based on a system linearization, as discussed in
\citet[Prop. 1]{Koehler2018a}. The condition can also be found in
Appendix~\ref{ass:sufficient_condition_locally_incremental_stabilizable}
and provides explicit choices for $V$ and $\pi$. It is, however, important
to note that the final algorithm only requires \emph{existence} of the
policy $\pi$ and the corresponding function $V$, rather than their explicit form.

These concepts lead to a robustified version of the nominal predictive
safety filter defined in \eqref{eq:MPSF_opt} and Algorithm~2
(Figure~\ref{fig:model_predictive_safety}, right), where we
omit~\eqref{eq:MPSF_opt_uncertainty_constraints} in the case of uniformly
bounded errors \eqref{eq:uniform_error_set}. Assumption~\ref{ass:local_incremental_stabilizability}
ties the model uncertainty \eqref{eq:uniform_error_set} to the constraint
tightening \eqref{eq:constraint_tightening} to ensure the existence of a safe backup plan
at all times and allows extension of the arguments for the nominal case
to a probabilistic model belief. If \eqref{eq:MPSF_opt} is feasible at time $k-1$
and the error bound $\mCS$ according to~\eqref{eq:uniform_error_set}, i.e.
$\max_{e\in\mCS}\mNorm{2}{e}$,
is sufficiently small with respect to $\epsilon$ (see also
Sections~\ref{subsec:design_parameters} for a detailed discussion)
with probability $p_\mSafe$,
then at time $k$, the input sequence based on the plan computed at
time step $k-1$
\begin{align}
	\mVpred_{i|k}\mDef \pi(\mZpred_{i|k},\mZpredOpt_{i+1|k-1,N},\mVpredOpt_{i+1|k-1,N})
\end{align}
for $i\in\mIntInt{0}{N-2}$ with $\mZpred_{0|k}=x(k)$, $\pi$ according
to Assumption~\ref{ass:local_incremental_stabilizability}, and $\mZpred_{i|k}$ according
to \eqref{eq:MPSF_opt_dynamic_constraints}, provides a feasible solution
to \eqref{eq:MPSF_opt} with planning horizon $N-1$ (Algorithm 2, line 6)
with probability $p_\mSafe$. Again, the tracking policy $\pi$ is only used in order
to show that a solution to \eqref{eq:MPSF_opt} exists, but it is not needed for
implementation of the approach. The same argument holds true for
all $\bar k \in \mIntInt{k+1}{k + N-1}$ until the terminal set is reached (line 10),
which allows us to establish safety at all times similarly to the nominal case.
A formal proof will be given in the following for the more general case including a
constraint on model confidence.

\emph{Planning in confident subspaces:}
To reduce conservatism introduced by uniformly overbounding the uncertainty
in \eqref{eq:uniform_error_set}, a central novelty in the proposed safety filter
is the ability to restrict planning to regions in the state and input space $\ZZ_c$
(see also Figure~\ref{fig:model_predictive_safety}) where we are sufficiently confident
about the system dynamics. More precisely, we restrict predictions to subspaces
where the model error~\eqref{eq:state_input_dependent_model_error} is contained in a
pre-specified, reduced allowable error set of the form
\begin{align}\label{eq:admissible_error}
	\mCS^\gamma \mDef \{e\in\RR^n | a_\mCS(e) \leq \gamma\mOnes{n_\mCS}\}, \quad 0<\gamma\leq 1
\end{align}
with scaling factor $\gamma$, which allows easy adjustment of the maximum
error magnitude due to the relation $e\in\mCS^\gamma \Rightarrow \mNormSmall{2}{e}\leq(\sqrt{n_\mCS}/c_\mCS)\gamma$,
see proof of Lemma~\ref{lem:feasible_solution}. A simple approach would be to compute the region $\ZZ_c$
offline and add it as an additional state and input constraint, as was similarly done for the case of
linear dynamics with state dependent uncertainties by \citet{soloperto2018learning}.
However, it is in general difficult to compute $\ZZ_c$ analytically and in addition,
the set needs to be recomputed once the model belief~\eqref{eq:general_nonlinear_system},
\eqref{eq:dynamics_model} is
updated based on observed data. We therefore reformulate the requirement to
stay inside $\ZZ_c$ as an implicit constraint, avoiding the explicit computation
of $\ZZ_c$, and include it in the online predictive safety filter problem
\eqref{eq:MPSF_opt} using the following definition:

\begin{definition}\label{def:set_valued_model_confidence_map}
	A set-valued map $\mSetValuedConfidence{p_\mSafe}$ mapping states
	and inputs from $\RR^n \times \RR^m$ to subsets of $\mCS$ with $\mCS\subset\RR^n$
	is a \emph{set-valued model confidence map}
	associated with \eqref{eq:general_nonlinear_system}, \eqref{eq:dynamics_model},
	for a given $\bar \theta\in\RR^q$ at probability level $p_\mSafe>0$, if
	\begin{align}
		\Pr\left(
			e(k,\theta_{\mathcal R})\in\mathcal E_{p_s}(x(k),u(k))),~
			\begin{matrix}
				\forall k\in\mathcal I_{[0,\bar N]} \\
				\forall x(k)\in\RR^n \\
				\forall u(k)\in\RR^n
			\end{matrix}
		\right)\geq p_s
		\label{eq:set_valued_model_confidence_map}	
	\end{align}
	holds, with $e(k,\theta_{\mathcal R})$ as defined in \eqref{eq:state_input_dependent_model_error}.
\end{definition}
While $\mCS$ must include all sufficiently common model errors, 
$\mSetValuedConfidence{p_\mSafe}(x,u)$ must only include errors that are sufficiently
common at $(x,u)$.
Note that according to Definition~\ref{def:set_valued_model_confidence_map}
it is not sufficient to guarantee that \eqref{eq:set_valued_model_confidence_map}
holds for some $k$, but it has to hold for all $0 \leq k\leq \bar N$
to ensure safety for all times, including also the case $\bar N\rightarrow\infty$.
In practice, it might be challenging to select a parametric system class and to infer
a representative parameter distribution $p(\theta | \mData)$ from a data set $\mData$ that
allows construction of a set-valued model confidence map. In the following we
therefore briefly discuss an example of how to design~\eqref{eq:set_valued_model_confidence_map}
from data using Bayesian regression and refer to \citet[Section 3]{Hewing2020} for a review of
data-driven prediction models that provide bounds on the model uncertainty.

\emph{Data-driven set-valued model confidence map:}
	Consider a Bayesian description of \eqref{eq:general_nonlinear_system}
	with prior distribution $p(\theta)$ and
	posterior estimate $ \theta \sim p(\theta|\mData)$, 
	inferred from available system data $\mData\mDef\{(x_i,u_i), f(x_i,u_i;\theta_{\mathcal R})\}_{i=1}^{N_\mData}$.
	Define a confidence region $\mConfidenceSet{p(\theta|\mData)}{p_\mSafe}$
	at probability level $p_\mSafe>0$ of the random parameters $\theta$ as
	$\Pr(\theta\in\mConfidenceSet{p(\theta|\mData)}{p_\mSafe})\geq p_\mSafe$.
	A set-valued model confidence map according to
	Definition \ref{def:set_valued_model_confidence_map} is then given by
	\begin{align}\label{eq:parametric_set_valued_confidence_map}
		&\mSetValuedConfidence{p_\mSafe}(x,u)= \\\nonumber
		&\{e\in\RR^n|e=f(x,u,\theta) - f(x,u; \bar \theta)~,
			\theta \in \mConfidenceSet{p(\theta|\mData)}{p_\mSafe} \},
	\end{align}
	as it follows from the definition of $\mConfidenceSet{p(\theta|\mData)}{p_\mSafe}$
	that
	\begin{align*}
    	\Pr(\star) & 
            & \geq \underbrace{\Pr(\star | \theta_{\mathcal R} \in \mathcal C_{p_s}(p(\theta|\mathcal D)))}_{=1 \text{ by definition of } \mathcal E_{p_s}}
                \underbrace{\Pr( \theta_{\mathcal R} \in \mathcal C_{p_s}(p(\theta|\mathcal D)))}_{\geq p_s}
	\end{align*}
	with $\star$ as shorthand for the random event introduced in
	\eqref{eq:set_valued_model_confidence_map}.
	Note that similar set-valued model confidence maps can be obtained
	when using non-parametric Gaussian process regression, by assuming that the
	system dynamics \eqref{eq:general_nonlinear_system} have bounded norm in
	a reproducing kernel Hilbert space \citep[Theorem~2]{Chowdhury2017}.
	In case of large amounts of available data on the whole state and input
	space a uniform confidence map can be selected using, e.g., Lipschitz arguments similar
	to~\citet{limon2017learning}, i.e.,
	$\forall x\in \XX, ~u\in\UU :~
		\mSetValuedConfidence{p_\mSafe}(x,u) \subseteq \mCS^1$,
	reducing to the special case \eqref{eq:uniform_error_set}.
	\END
	 
As discussed for the case of uniformly bounded errors, the tightened
constraints~\eqref{eq:constraint_tightening} ensure safety, if~\eqref{eq:uniform_error_set}
holds for $\max_{e\in\mCS^\gamma}\mNorm{2}{e}$ small enough. Since $e(k,\theta_{\mathcal R})$ is
unknown, and we cannot simply impose $e(k,\theta_{\mathcal R})\in\mCS^\gamma$ in~\eqref{eq:MPSF_opt}
to restrict planning to confident subsets, we make use of the model confidence map in
Definition~\ref{def:set_valued_model_confidence_map} to enforce 
\begin{align}\label{eq:non_tightened_uncertainty_constraint}
	\mSetValuedConfidence{p_\mSafe}(x(k), u(k)) \subseteq \mCS^\gamma,
\end{align}
implying $e(k,\theta_{\mathcal R})\in\mCS^\gamma$ with probability $p_\mSafe$.
To this end, we impose \eqref{eq:MPSF_opt_uncertainty_constraints} on the nominal plan
$\{\mZpred_{i|k}\}$, $\{\mVpred_{i|k}\}$, where constraint \eqref{eq:non_tightened_uncertainty_constraint} 
is tightened similarly to \eqref{eq:constraint_tightening} using
\begin{align}\label{eq:constraint_tightening_set_valued_map}
	\bar \mCS_i^\gamma  \mDef\{ e \in \RR^n | a_\mCS(e) \leq \gamma(1-\epsilon_i)\mOnes{n_\mCS}\}.
\end{align}
The tightening again ensures the existence of a feasible solution when replanning
with a shorter horizon (Algorithm 2, line 6).
In order for the filter to ensure safety in probability using
\eqref{eq:MPSF_opt_uncertainty_constraints},
we additionally require that small changes of the nominal predicted
trajectory must not lead to arbitrary large changes in the model confidence by
assuming that the set-valued model confidence map is Lipschitz continuous
in terms of the Hausdorff metric (see Definition~\ref{def:hausdorff_metric} and
\ref{def:Lipschitz_continuous_mapping} in
the appendix).
\begin{assumption}
	\label{ass:set_valued_model_confidence_map}
	There exists a set-valued model confidence map $\mSetValuedConfidence{p_\mSafe}$
	associated with \eqref{eq:general_nonlinear_system},
	\eqref{eq:dynamics_model}, which is Lipschitz continuous
	with Lipschitz constant $L_{\mSetValuedConfidence{p_\mSafe}}$
	under the Hausdorff metric with respect
	to $d_{\mathcal \RR^m}(a,b)\mDef\mNorm{2}{a-b}$.
\end{assumption}
Note that for common models, such as Gaussian Processes,
Assumption~\ref{ass:set_valued_model_confidence_map} is generally fulfilled,
compare also with \citet[Proposition 11]{Fisac2019}.
The above assumptions allow for extension the ideas from the uniform error
bound to make use of a potentially reduced error bound that is ensured by imposing
\eqref{eq:MPSF_opt_uncertainty_constraints} on the backup plan, and thereby
again characterize the relation between the tightening
$\epsilon$ in \eqref{eq:constraint_tightening},\eqref{eq:constraint_tightening_set_valued_map} and
the specified tolerated model error~\eqref{eq:admissible_error}.
This leads us to the main result of the paper, showing that the proposed predictive
safety filter guarantees safety in probability at all times according to~\eqref{eq:chance_constraints}.
\begin{theorem}\label{thm:PSF}
	Let Assumptions~\ref{ass:terminal_safe_set}, \ref{ass:local_incremental_stabilizability} and
	\ref{ass:set_valued_model_confidence_map} hold and select a tightening factor $\epsilon > 0$.
	If $L_{\mSetValuedConfidence{p_\mSafe}}$ in Assumption~\ref{ass:set_valued_model_confidence_map}
	is sufficiently small, i.e. if $L_{\mSetValuedConfidence{p_\mSafe}} \leq c \epsilon$ for a
	sufficiently small constant $c>0$,
	then one can always select a sufficiently small $\gamma > 0$ such that the initial
	feasibility of \eqref{eq:MPSF_opt} for $x(0)$ implies that $u(k)=\pi_\mSafe(k, x(k), u_\LL(k))$
	as defined in Algorithm~2 ensures safe system operation	according to~\eqref{eq:chance_constraints}.
\end{theorem}
The proof is provided in the appendix.
Theorem~\ref{thm:PSF} implies that for sufficiently small
$L_{\mSetValuedConfidence{p_\mSafe}}$ one can specify a constraint tightening
through $\epsilon$ and impose a corresponding sufficiently small
admissible error set scaling $\gamma$ in~\eqref{eq:admissible_error}, such
that if \eqref{eq:MPSF_opt} is initially feasible for $x(0)$, application
of Algorithm~2 will keep the system safe in probability. Thereby, the upper bound on
$L_{\mSetValuedConfidence{p_\mSafe}}$ results from the linear lower bound on
$a_\mCS$ in~\eqref{eq:uniform_error_set}, \eqref{eq:admissible_error} and intuitively
means that the set-valued model confidence map estimate in
Assumption~\ref{ass:set_valued_model_confidence_map} can only
change at a specific rate with changing states or inputs that are linearly
bounded in terms of the tightening fraction $\epsilon$.

While the exact values of the bounds
derived in the proof of Theorem~\ref{thm:PSF} might be
difficult to compute explicitly for design of the PSF,
the corresponding analysis in Appendix~\ref{app:proof_of_theorem}
unveils inner relations of all design parameters that
can be used for efficient practical tuning guidelines as presented
in Section~\ref{subsec:design_parameters}. A specific choice of parameters
can then be verified as described in Appendix~\ref{app:design_certification}.
\begin{remark}
	While the combination of the proposed safety filter
	with a learning-based controller naturally restricts exploration,
	the probabilistic model together with probabilistic constraints
	provide a principled way to adjust the probability associated
	with the confident subset, and thereby allow for some
	exploration beyond the available data as illustrated in the numerical
	example in Section~\ref{sec:numerical_example}. Large model uncertainties
	might, however, cause infeasibility of the PSF problem~\eqref{eq:MPSF_opt}
	at the initial condition of the system. This would either require
	the enlargement of the prediction horizon $N$ or the lowering of the
	probability level for safety $p_\mSafe$ in~\eqref{eq:chance_constraints},
	see also Section~\ref{subsec:design_parameters} for practical tuning guidelines.
\end{remark}
\subsection{PSF design parameters}\label{subsec:design_parameters}
In the following, we provide a more detailed discussion of the design parameters
and how to select them.
\par
$\rho\in (0,1)$:  Minimum contraction rate (`speed') at which the
system can reduce the distance (in terms of an appropriate energy function)
to a nominal reference trajectory. For example, consider the extreme case
of a deadbeat controller that can steer the system to any given reference
in one time step. According to Assumption~\ref{ass:local_incremental_stabilizability}
this translates into $\rho \approx 0$, which renders the constraint
tightening~\eqref{eq:constraint_tightening} constant after one time step.
In contrast, systems with very slow convergence rates are characterizes
with $\rho \approx 1$, corresponding to the worst-case in
terms of the constraint tightening~\eqref{eq:constraint_tightening}.
A cautious choice is therefore $\rho\mDef 0.9\bar 9$.

$\epsilon > 0$: Constraint tightening factor along predicted
backup plans. While $\rho$ depends on intrinsic
system properties, $\epsilon$ is a design parameter that
allows a trade off of the maximum tolerated prediction model
errors~\eqref{eq:admissible_error}, i.e. the magnitude
of $\gamma$ against the conservatism of the predictive safety filter,
i.e. the constraint tightening. This can be seen explicitly through the sufficient
bounds on $\gamma$ in the proof of~Lemma~\ref{lem:feasible_solution},
\eqref{eq:app_maximum_error_magnitude}, which are linear in $\epsilon$, i.e.
$\gamma \leq c_{\gamma} \epsilon$. To satisfy the lower bound according to
Theorem~\ref{thm:PSF} while preventing the tightened sets~\eqref{eq:constraint_tightening}
from being empty at the end of the planning horizon, a cautious initial choice is given by
$\epsilon \leq (1-\sqrt{\rho})/(1-\sqrt{\rho}^N)$, where the prediction horizon length $N$
can additionally be reduced to account for larger values of $\epsilon$.

$\gamma > 0$: From the definition of $\mCS^\gamma$ in
\eqref{eq:uniform_error_set} and \eqref{eq:admissible_error} it follows
that $\max_{e\in\mCS^\gamma}\mNorm{2}{e}\leq \gamma\sqrt{n_\mCS}/c_\mCS$ holds, i.e.
$\gamma$ linearly affects the maximum allowable uncertainty
in the confident subset of the state space.
From Theorem~\ref{thm:PSF} it follows that for any valid $(\rho, \epsilon)$
a $\gamma > 0$ exists, such that initial feasibility of the predictive safety
filter implies chance constraint satisfaction according to~\eqref{eq:chance_constraints}.
The bound on $\gamma$ is provided in~\eqref{eq:app_maximum_error_magnitude} in the appendix.
Since smaller values of $\gamma$ render the set-valued model confidence map
constraint~\eqref{eq:MPSF_opt_uncertainty_constraints} more conservative,
the goal during tuning is to find the largest tolerable uncertainty
$\gamma$ for a given configuration $(\rho, \epsilon)$.

$p_\mSafe \in [0, 1]$: Desired probability level of safety according to
\eqref{eq:chance_constraints}. Depending on the application, one might
consider lowering the probability level $p_\mSafe$ for an efficient
exploration phase, before enforcing larger values $p_\mSafe$ for cautious
long term operation. More precisely, the limit case $p_\mSafe= 0$ allows selection of
$\mSetValuedConfidence{p_\mSafe} = \{0\}$, which virtually disables the
set-valued model confidence map~\eqref{eq:set_valued_model_confidence_map}
and backup plans are not restricted to confident subsets anymore. In turn,
selecting $p_\mSafe \approx 1$ results in a robust version of the
predictive safety filter and therefore limits exploration.

In summary, the small number of design parameters and their interpretability
allow for an efficient design of the PSF without more involved
and potentially conservative design procedures to formally satisfy the required
assumptions, e.g., Assumption~\ref{ass:local_incremental_stabilizability}
(see also~\citet{Koehler2018a}) or Assumption~\ref{ass:set_valued_model_confidence_map}.
A cautious initial selection of the design parameters is given by $\rho\approx 1$,
$\epsilon\approx N^{-1}$, and $\gamma$ possibly small for a required probability
level $p_\mSafe$ and planning horizon $N$. A practical choice of $\mCS^\gamma$ is discussed
in Section 5. The set of parameters can then be verified offline as described in
Appendix~\ref{app:design_certification}. If these conservative design
parameters cannot be verified, then either the planning horizon $N$ can be reduced to
increase $\epsilon$, Assumption~\ref{ass:local_incremental_stabilizability}
does not hold, more data needs to be collected, or the prior information about $\theta$
needs to be refined to render the set-valued model model confidence map according
to Assumption~\ref{ass:set_valued_model_confidence_map} less conservative.

\section{Application to numerical examples}\label{sec:numerical_example}
\subsection{Swing-up: Safe exploration beyond initial data}\label{subsec:pendulum}
\begin{figure}
	\centering
	\includegraphics[width=0.47\textwidth]{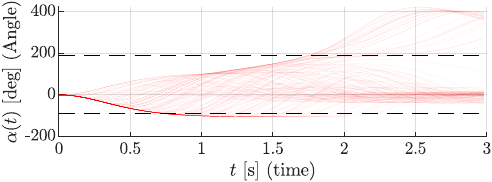}
	\includegraphics[width=0.47\textwidth]{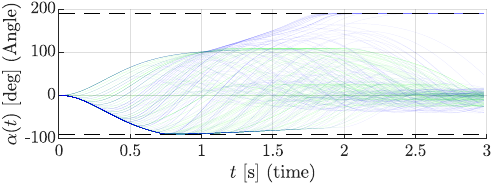}
	\includegraphics[width=0.47\textwidth]{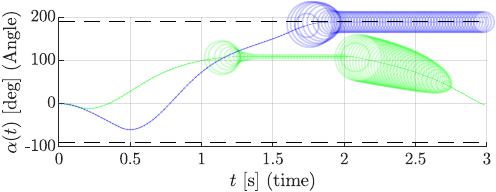}
	\caption{
		Comparison of closed-loop swing-up trajectories during 120 learning episodes 
		under challenging position constraints (dashed lines).
		\emph{Top:} Red lines show different learning episodes based on
		an unsafe learning policy.
		\emph{Middle:} Closed-loop learning trajectories using
		a predictive safety filter with $10$ data points (green) and $18000$ data points
		(blue).
		\emph{Bottom:} Resulting safe optimal closed-loop performance.
		The circle radii indicate the relative magnitude of safety ensuring
		modifications of the learning policy.}
	\label{fig:inverted_pendulum_trajectory}
\end{figure}
We consider the classical
control problem of swinging up a pendulum from the downward
position with angle $\alpha=0~\mathrm{[deg]}$ to the upward position
($\alpha=180~\mathrm{[deg]}$) with limited input authority,
unknown system parameters, and under challenging safety
constraints of the form $-90~\mathrm{[deg]} \leq \alpha \leq 190~\mathrm{[deg]}$, such that
the pendulum is not allowed to tip over, once the upward position has been reached.
The discretized dynamics $x(k+1)=f(x(k),u(k))$ are simulated using $x_1(k+1) = x_1(k) + hx_2(k)$ and
$x_2(k+1) = x_2(k) -\frac{hg}{l}\sin(x_1(k)) - \frac{h\eta}{ml^2}x_2(k) + \frac{h}{ml^2} u(k)$,
where $x_1(k)=\alpha(k)$ is the angle, $x_2(k)=\dot\alpha(k)$ is the angular velocity at time step
$k$, $h=0.02~\mathrm {[s]}$ is the discretization interval, $g=9.81~\mathrm{[m/s^2]}$ is the
gravity constant, $l=0.5~\mathrm{[m]}$ is the length, $m=0.15~\mathrm{[kg]}$ is the mass,
$\eta=0.1~\mathrm{[Nms / rad]}$ is the friction and the input torque $u$ is restricted to
$|u|\leq 0.7~\mathrm{[Nms / rad]}$.\\
\emph{Unsafe learning policy $\pi_\LL$:}
For learning the swing-up task, we consider an episodic learning
setting with horizon length $\bar N=120$ and parametrize a bang-bang open-loop
input signal as
\begin{align}
	\label{eq:learning_based_input_signal}
	\pi_\LL(k;k_{s_1},k_{s_2}) = 
		\begin{cases}
			-0.7, \quad & k\leq k_{s_1},\\
			0.7, \quad & k_{s_1}\leq k\leq k_{s_2}, \\
			0, & \text{else},
		\end{cases}
\end{align}
with switching times $k_{s_1}$ and $k_{s_2}$ subject
to $k_{s_1}\in [0,k_{s_2}]$ and $k_{s_2}\in[0, 150]$.
The learning objective is defined as $\ell(x,k) = (\alpha(k) - \pi)^2
	+ (\mVpredOpt_{0|k}-\pi_\LL(k;k_{s_1},k_{s_2}))^2$, where the first term
describes the distance to the desired upward position given
by $180~\mathrm{[deg]}$, while the second term penalizes safety-ensuring
interventions by the safety filter and therefore accelerates learning convergence
as discussed in~\citet{Akametalu2014}. Efficient learning-based optimization
of the parameters $k_{s_1}$ and $k_{s_2}$ is performed
using Bayesian Optimization as described in~\citet{Neumann-Brosig2019},
selecting parameter configurations that
automatically trade-off exploration of the parameter space and exploitation
of promising subsets. While direct application of this learned policy
yields a swing-up after some episodes, it causes significant constraint
violations as shown in Figure~\ref{fig:inverted_pendulum_trajectory}~(top),
motivating the application of the presented safety filter in the following.

\emph{Predictive safety filter from data:}
The transition model~\eqref{eq:general_nonlinear_system},
\eqref{eq:dynamics_model} is obtained via
linear Bayesian regression~\citep{rasumssen2006Gaussian},
i.e., $f(x,u) = \theta^\top \phi(x)$, with $\phi(x,u)=[x_1,x_2,\sin(x_1),u]^\top$, unknown
parameters $\theta \in\RR^{2\times 4}$ with Gaussian priors, and Gaussian noise on
obtained system measurements.
The set-valued model confidence map according to Definition~\ref{def:set_valued_model_confidence_map}
for parametric uncertainties is given by~\eqref{eq:parametric_set_valued_confidence_map} and
can here be defined as $\mSetValuedConfidence{p_\mSafe}(x,u) = \{e\in \RR^n |
	e^\top \Sigma^{-1}(x,u) e \leq \chi_2^2(p_\mSafe)\}$ where
$\Sigma(x,u) = \mathrm{diag}((\sigma_{i}^2(x,u))_{i=1,2})$ with $\sigma_{i}^2(x,u)$ being the
posterior variance of $f_i(x,u)$ conditioned on data and $\chi_2^2(p_\mSafe)$ the chi-squared
distribution of degree $2$ \citep{Slotani1964}.
Starting with 10 data points around the downward position, we update the model belief
using the acquired data after each episode. 
The tightening was experimentally chosen to $\rho=0.999,\epsilon=0.02$
using sampled system realizations of the posterior distribution as described
in Section~\ref{app:design_certification}.
The corresponding admissible error set $\mCS^\gamma$ is defined as the $2$-norm ball
with radius $\gamma = 0.02$. Consequently, set-valued map constraints of the form
$\mSetValuedConfidence{p_\mSafe}(x,u) \subseteq \bar\mCS^\gamma_i$ according to \eqref{eq:MPSF_opt_uncertainty_constraints}
can be efficiently implemented as $\sqrt{\sigma_{f_j}^2(x,u)\chi_2^2(p_\mSafe)} \leq (1-\epsilon_i)0.02$ for $j=1,2$,
i.e. by enforcing all semi-axes of $\mSetValuedConfidence{p_\mSafe}(x,u)$
to be smaller or equal than the radius of the admissible error set $\bar\mCS^{0.02}_i$.
The desired probability of chance constraint satisfaction $p_\mSafe$ was chosen as $0.95$.
As the terminal safe set, we select
$\mSafe^t \mDef \{ \alpha, \dot\alpha| -30~\mathrm{deg} \leq \alpha \leq 30~\mathrm{ [deg]},
~|\dot\alpha|\leq 30~\mathrm{ [deg/sec]}$ with $\pi_\mSafe^t = 0$. The resulting 
problem \eqref{eq:MPSF_opt} with planning horizon $N=50$ was solved in real-time using
Ipopt \citep{wachter2006implementation} together with the CasADi framework~\citep{Andersson2018}
for automatic differentiation.

\emph{Results:}
Combining the learning-based swing up policy with a predictive safety
filter based on only $10$ initial data points around the stable downward position
results in cautious closed-loop system trajectories that are displayed as green lines in
Figure~\ref{fig:inverted_pendulum_trajectory} (middle).
The corresponding optimal solution after 120 learning episodes is depicted in
Figure~\ref{fig:inverted_pendulum_trajectory} (middle). We can then leverage the
cumulated data from the first experiment ($18000$ data samples) to refine
the prediction model of the safety filter. The additional data enables a significantly less conservative
learning behavior, see blue trajectories in Figure~\ref{fig:inverted_pendulum_trajectory} (middle),
and supports a complete swing-up, which demonstrates safe exploration beyond available
data. The corresponding optimal solution after 120 learning episodes is shown in
Figure~\ref{fig:inverted_pendulum_trajectory} (middle), where the circle radii
indicate the magnitude of safety ensuring modifications of the learning policy.

\subsection{Safe data-driven quadrotor learning control}\label{subsec:quadrotor}
\begin{figure}
	\centering
	\includegraphics[width=0.178\textwidth,valign=t]{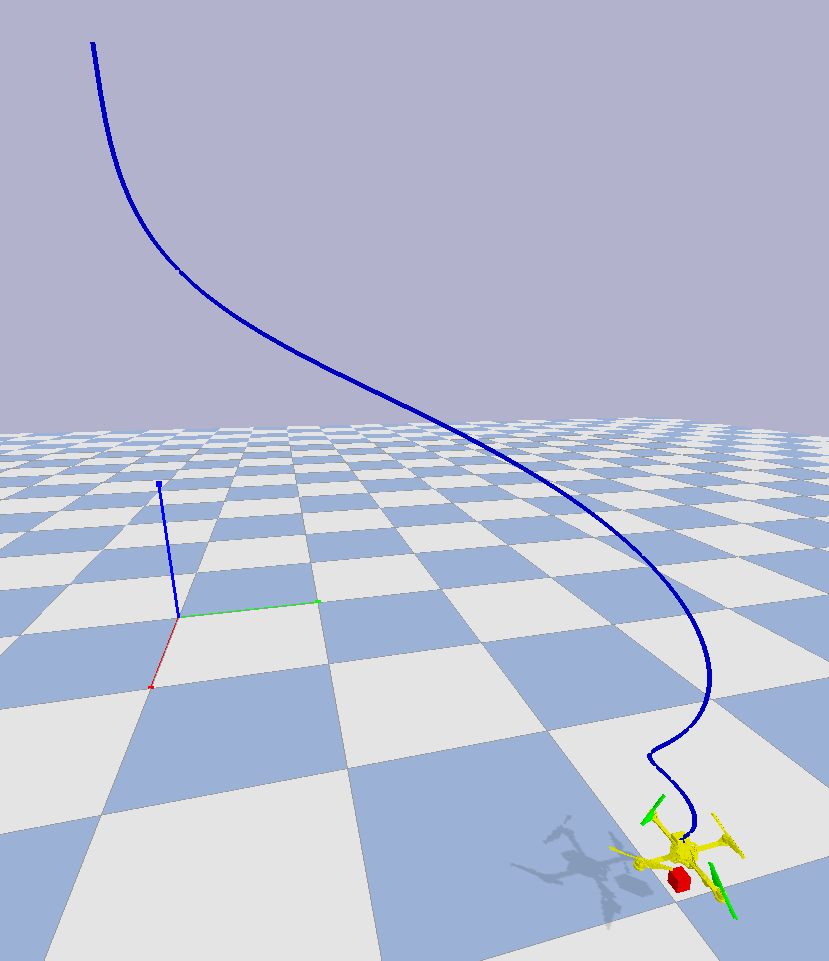}
	\hfill
	\includegraphics[width=0.27\textwidth,valign=t]{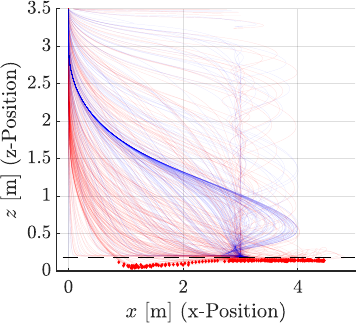}
	
	\vspace{.5cm}

	\includegraphics[width=0.48\textwidth]{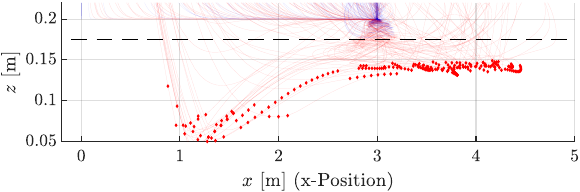}
	\caption{
	Quadrotor experiment using the Bullet Physics SDK
		\citep{coumans2019}. \emph{Top-left:} Graphical interface
		showing the optimal safe trajectory (blue line).
		\emph{Top-right:} Quadrotor trajectories projected on
		the $x-z$ plane using an unsafe policy search (red lines)
		and the safety augmented policy search (blue lines).
		\emph{Bottom:} Zoom-in of top-right plot, where
		red dots represent states with $<0.01~\mathrm{[m]}$
		minimum distance to the ground, which we classify as
		ground contact.
		}
	\label{fig:quadrotor}
\end{figure}
To demonstrate the presented method for a more challenging simulation
example, we consider the AscTec Hummingbird drone, simulated in the
Bullet Physics SDK \citep{coumans2019}, see Figure~\ref{fig:quadrotor} (Top-left),
in combination with a single rotor force model~\citep{Furrer2016}.
We employ a two-layer control structure, where an inner PD
control loop takes desired pitch, roll and vertical acceleration in the body frame
and outputs control signals to the motors.
This enables modeling of the inner controlled system around the hovering equilibrium
as done in~\citet{Hu2018} using 10 states $\tilde x\in\RR^{10}$, three inputs $u\in\RR^3$,
and dynamics of the form $\tilde x(k+1)=\theta^\top \phi(\tilde x,u)$.
State constraints are given by minimum height $z\geq 0.175~\mathrm{[m]}$
and maximal vertical velocity $|\dot z|\leq 1~\mathrm{[m/s]}$ to avoid ground
contact and to ensure the validity of the dynamics model. The inputs to
the inner control loop are normalized to $|u_i|\leq 1$.

\emph{Unsafe learning policy $\pi_\LL$:}
The overall goal is to approach the landing position
$x=3~\mathrm{[m]}$, $y=2~\mathrm{[m]}$, and
$z=0.2~\mathrm{[m]}$ close to the ground, which is indicated
by the red cube in Figure~\ref{fig:quadrotor} (Top-left),
starting from an initial hovering position
at $z=3.5~\mathrm{[m]}$, $x=y=0~\mathrm{[m]}$ using
an outer PD controller with state and input saturation, which is
parametrized as
\begin{align*}
	\pi_\LL(\tilde x;p,d) \mDef
	\begin{cases}
		\mathrm{clip}(p_{12} (x_d-x) + d_{12} \dot x,-1,1) \\
		\mathrm{clip}(p_{12} (y_d-y) + d_{12} \dot y,-1,1) \\
		\mathrm{clip}(p_{3} (z_d-z) + d_3 \dot z,-1,1)
	\end{cases}
\end{align*}
where $\mathrm{clip}(x,c_1,c_2)\mDef \max(\min(x,c_2),c_1)$ and
with PD-controller gains $p_{12},p_3\in [0, 10]$
and $d_{12}, d_3 \in [-10,0]$. Similar to the example
presented in Section~\ref{subsec:pendulum}, we use
Bayesian Optimization~\citep{Neumann-Brosig2019},
with
$\ell(\tilde x, u)\mDef |x_d-x| + |y_d-y| + |z_d-z| + 100 \mNormSmall{2}{\pi_\LL(\tilde x) - \mVpredOpt_{0|k}}$
that heavily penalizes safety ensuring actions by the safety filter.
Direct application of the learning algorithm yields
a significant number of ground contacts
(in addition to violations of the maximum vertical velocity) as
shown in Figure~\ref{fig:quadrotor} by the red trajectories over a total of
240 learning episodes, where ground contacts are defined as $<0.01~\mathrm{[m]}$
minimum distance to the ground and are highlighted by red dots.

\emph{Predictive safety filter from data:}
Similarly to the numerical example in Section~\ref{subsec:pendulum},
we construct a model via Bayesian Regression using a Gaussian prior on
the parameters $\theta$ together with Gaussian noise on observations.
The data required for infering the prediction model is generated
through experiments at a safe altitude using 100 random step
inputs that are applied for $60~\mathrm{[s]}$
to the inner control loop. Based on the inferred model,
the constraint tightening was experimentally chosen as
$\rho=0.999, \epsilon=0.01$ using posterior samples as
described in Section~\ref{app:design_certification} with a maximum
allowable error set $\mCS^\gamma$ defined as the $2$-norm ball
with radius $0.02$. The planning horizon is given by $N=20$
and the terminal set is selected as a sufficiently high altitude
of $z\geq 1.5~\mathrm{[m]}$, from which we can guarantee
constraint satisfaction for all future times using, e.g.,
suboptimal PD controller gains $p_{12} = p_3 = 0.5$ and
$d_{12} = d_3 = -0.5$.
The set-valued model confidence map is set up analogue to
Section~\ref{subsec:pendulum} with desired chance constraint
satisfaction $p_\mSafe = 0.9$.

\emph{Results:}
As shown in Figure~\ref{fig:quadrotor} the predictive safety
filter enables constraint satisfaction during all 240 learning-episodes
and results in a favorable optimal trajectory as shown in
Figure~\ref{fig:quadrotor}~(Top-left), for which
the safety filter is permanently inactive as specified via
the objective function.
\section{Conclusion}
This paper has addressed the problem of safe
RL by introducing a predictive
safety filter, which enables modularity in terms of
safety and the employed RL algorithm.
An optimization-based formulation was proposed
that provides rigorous safety guarantees using
a possibly data-driven approximate system model. 
By its capability to consider nonlinear and complex
system descriptions without being overly conservative,
we believe that the proposed approach is an important
step towards safe RL for realistic applications.

\bibliographystyle{dcu}
\bibliography{bibliography.bib}

\appendix
\section{Appendix}
\subsection{Lipschitz continuity w.r.t. Hausdorff metric}
\begin{definition}\label{def:hausdorff_metric}
	The \emph{Hausdorff metric} between two sets $\mathcal A$ and $\mathcal B$ in a metric
	space $(M,d_{\mathrm M})$ is defined as
	\begin{align*}
		\mHausdorff{\mathcal A}{\mathcal B} \mDef \max \left\{
		\adjustlimits
		\sup_{a\in\mathcal A} \inf_{b\in\mathcal B} d_{\mathrm M}(a,b),
		\adjustlimits
		\inf_{a\in\mathcal A} \sup_{b\in\mathcal B} d_{\mathrm M}(a,b)
		\right\}.
	\end{align*}
\end{definition}

\begin{definition}\label{def:Lipschitz_continuous_mapping}
	A set valued map $\mCS$ mapping vectors from 
	$\mathcal A \subseteq \RR^n$ to subsets of
	$\RR^m$ is called Lipschitz continuous with Lipschitz constant 
	$L_{\mCS}\geq 0$ under the Hausdorff metric with respect
	to the 2-Norm, if for all $a, b\in\mathcal A$ it holds that
	$
		\mHausdorff{\mathcal E(a)}{\mathcal E(b)}\leq L_\mCS \mNorm{2}{a-b}.
	$
\end{definition}

\subsection{Proof of Theorem \ref{thm:PSF}}\label{app:proof_of_theorem}

We begin by deriving a bound on the amount at which small changes in the planned
nominal trajectory $\{\mZpredOpt_{i|k}, \mVpredOpt_{i|k}\}$
affect the set membership constraint \eqref{eq:MPSF_opt_uncertainty_constraints}
in Lemma~\ref{lem:containment_following_confidence_set}.
Based on Assumption~\ref{ass:local_incremental_stabilizability} together
with Lipschitz continuity of the state, input, and terminal constraints,
as well as the aforementioned bound on the set membership constraint,
we then show that feasibility of~\eqref{eq:MPSF_opt} for planning horizon
$N$ at time $k$ together with $e(k,\bar\theta)\in\mSetValuedConfidence{p_\mSafe}(x(k),u(k))$
implies the existence of a feasible solution at time $N-1$
in Lemma~\ref{lem:feasible_solution}.
Finally, we iteratively apply Lemma~\ref{lem:feasible_solution},
to prove Theorem~\ref{thm:PSF}.

In the following, we consider a model error set
$\mCS\mDef\{e\in\RR^n | a_\mCS(e) \leq \mOnes{n_\mCS}\}$, a safe terminal
set $\mSafe^t\mDef\{ x \in \RR^n | a_\mSafe(x) \leq \mOnes{n_\mSafe}\}$
according to Assumption~\ref{ass:terminal_safe_set} 
with $\mDefFunction{a_{\mCS,i}}{\RR^n}{\RR}$ and $\mDefFunction{a_{\mSafe,i}}{\RR^n}{\RR}$,
both Lipschitz continuous functions with constants $L_{a_{\mCS}}$, $L_{a_\mSafe}$.
In the predictive safety filter optimization problem \eqref{eq:MPSF_opt}
the constraints are defined according to the tightening
\eqref{eq:constraint_tightening} and \eqref{eq:constraint_tightening_set_valued_map}.
We denote an optimal solution of \eqref{eq:MPSF_opt} at time $k$
with planning horizon $N$ as nominal input sequence
$\mVpredOpt_{0|k,N},..,\mVpredOpt_{N-1|k,N}$ with corresponding
nominal state sequence $\mZpredOpt_{0|k,N},..,\mZpredOpt_{N|k,N}$
\eqref{eq:MPSF_opt_dynamic_constraints}.

\begin{lemma}\label{lem:containment_following_confidence_set}
	Let Assumption~\ref{ass:set_valued_model_confidence_map} hold. If
	$x,x^+\in\XX$, $u,u^+\in\UU$ and
	$\mSetValuedConfidence{p_\mSafe}(x,u)\subseteq\mCS^\gamma$, then
	$\mSetValuedConfidence{p_\mSafe}(x^+, u^+) \subseteq \mCS^{\gamma,+}(\mNorm{2}{\Delta z})$,
	where
	$\Delta z\mDef \mNorm{2}{[x^\top,u^\top] - [x^{+\top}, u^{+\top}]}$
	and $\mCS^{\gamma,+}(\mNorm{2}{\Delta z})\mDef\{e\in\RR^n | a_\mCS(e) \leq \gamma\mOnes{n_\mCS} +
		L_{a_\mCS}L_{\mSetValuedConfidence{p_\mSafe}}\mNorm{2}{\Delta z} \mOnes{n_\mCS} \}$.
\end{lemma}
\begin{proof}
	The essential observation is that all
	$e^+\in\mSetValuedConfidence{p_\mSafe}(x^+, u^+)$
	can be written as
	$e^+ = e^* + \Delta e$ with
	$e^*\mDef \arginf_{e\in\mSetValuedConfidence{p_\mSafe}(x,u)}\mNorm{2}{e^+ - e}$, $\Delta e \in \RR^n$
	for which we have by Assumption~\ref{ass:set_valued_model_confidence_map}
	that
	\begin{align*}
		\mNorm{2}{\Delta e} & = \mNorm{2}{e^+ - e^*}
		=  \inf_{e\in\mSetValuedConfidence{p_\mSafe}(x,u)}{\mNorm{2}{e^+ - e}}                        \\
		                    & \leq \adjustlimits\sup_{e^+\in\mSetValuedConfidence{p_\mSafe}(x^+,u^+)}
		\inf_{e\in\mSetValuedConfidence{p_\mSafe}(x,u)}{\mNorm{2}{e^+ - e}}                           \\
		                    & \leq L_{\mSetValuedConfidence{p_\mSafe}} \mNorm{2}{\Delta z}.
	\end{align*}
	This allows us to derive
	\begin{align}
		a_\mCS(e^+) & = a_\mCS(e^*) + a_\mCS(e^+) - a_\mCS(e^*) \nonumber                                                         \\
		            & \leq a_\mCS(e^*) + L_{a_\mCS}\mNorm{2}{\Delta e}\mOnes{n_\mCS}     \nonumber                                 \\
		            & \leq a_\mCS(e^*) + L_{a_\mCS}L_{\mSetValuedConfidence{p_\mSafe}}\mNorm{2}{\Delta z} \mOnes{n_\mCS} \nonumber\\
					& \leq \gamma \mOnes{n_\mCS} + L_{a_\mCS}L_{\mSetValuedConfidence{p_\mSafe}}\mNorm{2}{\Delta z} \mOnes{n_\mCS}
					\label{eq:containment_following_confidence_set_proof_1}
	\end{align}
	since by definition $e^*\in\mSetValuedConfidence{p_\mSafe}(x,u)\subseteq\mCS^\gamma$.
	Therefore, for all $e^+\in\mSetValuedConfidence{p_\mSafe}(x^+, u^+)$,
	\eqref{eq:containment_following_confidence_set_proof_1} holds,
	which implies $\mSetValuedConfidence{p_\mSafe}(x^+, u^+) \subseteq \mCS^{\gamma,+}(\mNorm{2}{\Delta z})$,
	completing the proof.
\end{proof}
\begin{lemma}\label{lem:feasible_solution}
	Let Assumptions~\ref{ass:local_incremental_stabilizability} and
	\ref{ass:set_valued_model_confidence_map} hold.
	For every $\epsilon > 0$, there exist corresponding
	$c>0,\gamma > 0$ such that if
	1) $L_{\mSetValuedConfidence{p_\mSafe}} \leq c\epsilon$,
	2) Problem \eqref{eq:MPSF_opt} is feasible at time $k$ with prediction horizon $N>0$,
	3) $u(k)=\pi_\mSafe(k,x(k),u_\LL(k))=\mVpredOpt_{0|k}$ is applied to
		\eqref{eq:general_nonlinear_system}, and
	4) $e(k,\theta_{\mathcal R})\in\mSetValuedConfidence{p_\mSafe}(x(k),u(k))$ with
			probability $1$,
	then the input sequence
	\begin{align*}
		\mUpredCand_{i|k+1}\mDef \pi(\mZpredCand_{i|k+1},\mZpredOpt_{i+1|k,N},\mVpredOpt_{i+1|k,N})
		~\text{ for } i\in\mIntInt{0}{N-2}
	\end{align*}
	with $\pi$ according to Assumption~\ref{ass:local_incremental_stabilizability},
	corresponding nominal state sequence
	$\mZpredCand_{i+1|k+1} = f(\mZpredCand_{i|k+1},\mUpredCand_{i|k+1},\bar \theta)$ 
	with $\mZpredCand_{0|k+1}=x(k+1)$ $\forall ~ i\in\mIntInt{0}{N-2}$ and
	$\mZpredCand_{-1|k+1}=x(k)$, is a feasible solution to \eqref{eq:MPSF_opt} at time $k+1$
	with prediction horizon $N-1$.
\end{lemma}
\begin{proof}
	The following proof is a modified and extended Version of \citet[Proposition 5]{Koehler2018b},
	which considers nonlinear systems with additive disturbances of the form
	$x(k+1) = f(x(k),u(k)) + w(k)$, to address model \eqref{eq:general_nonlinear_system},
	\eqref{eq:dynamics_model} in combination with the set-valued model confidence
	and terminal safe set constraints \eqref{eq:MPSF_opt_uncertainty_constraints}
	and \eqref{eq:MPSF_opt_safety_constraint}.
	The proof makes use of the conditions in Assumption~\ref{ass:local_incremental_stabilizability}
	to derive bounds on the difference between the optimal plan at time $k$, and the
	constructed plan $\{\mZpredCand_{i|k+1}\}$, $\{\mVpredCand_{i|k+1}\}$ at time $k+1$,
	which is in turn used to show that the constraint tightening implies that
	the constructed plan is a feasible solution for \eqref{eq:MPSF_opt} with planning
	horizon $N-1$. In order to streamline notation, note that 
	Assumption~\ref{ass:local_incremental_stabilizability} holds for all $v,v^+\in\UU$
	which allows us, together with the fact that $\mVpredOpt_{i|k}\in\UU$
	for all $i\in\mIntInt{0}{N-1}$ due to feasibility at time $k$, to omit the third argument
	of $V$ in the following analysis.
	
	We start by bounding errors $e \in \mCS^\gamma$ in terms of the
	scaling factor $\gamma$ as defined in~\eqref{eq:admissible_error}. Using
	the lower bound $\mNorm{2}{a_{\mCS}(e)}\geq c_\mCS \mNorm{2}{e}$ we obtain
	the relation $e\in\mCS^\gamma \Rightarrow a_\mCS(e)\leq \gamma \mOnes{n_\mCS}\Rightarrow
	\Vert a_\mCS(e) \Vert_\infty \leq \gamma$, and it follows that
	\begin{align}\label{eq:error_gamma_bound}
		\text{for all }e\in\mCS^\gamma \text{ it holds }\mNorm{2}{e} \leq c_\gamma\gamma
	\end{align}
	with $c_\gamma \mDef (\sqrt{n_\mCS}/c_\mCS)$.
	In the next step, we derive a bound $\gamma_1 > 0$ on
	$\gamma$ such that $V(x(k+1),\mZpredOpt_{1|k,N})\leq \delta$ holds. Select
	$\gamma \leq \gamma_1 \mDef c_\gamma^{-1}\sqrt{\frac{\delta}{c_u}}$
	and note that by assumption and constraint \eqref{eq:MPSF_opt_uncertainty_constraints}
	we obtain
	\begin{align}\label{eq:lemma_feasible_solution_proof_0}
		\mNorm{2}{x(k+1)-\mZpredOpt_{1|k,N}}^2 \leq c_\gamma^2 \gamma^2 \leq \frac{\delta}{c_u}
	\end{align}
	and combined with Assumption~\ref{ass:local_incremental_stabilizability} it follows
	\begin{align}\label{eq:lemma_feasible_solution_proof_1}
		V(x(k+1),\mZpredOpt_{1|k,N}) & \leq c_u \mNorm{2}{x(k+1)-\mZpredOpt_{1|k,N}}^2\leq \delta.
	\end{align}
	Next, we show when $\gamma>0$ is small enough, $\mVpredCand_{i|k+1}$ is
	a feasible candidate input sequence to \eqref{eq:MPSF_opt} with planning horizon $N-1$ in two steps.
	In \emph{Step 1} we show that $\mZpredCand_{i|k+1}\in\bar \XX_i$ holds for all
	$i\in\mIntInt{0}{N-2}$ by induction, which allows us in \emph{Step 2} to construct sufficient
	bounds on $\gamma$ that imply feasibility via the tightening sequence
	\eqref{eq:constraint_tightening_sequence} of the remaining constraints
	in~\eqref{eq:MPSF_opt}.\\
	\emph{Step 1:}
	For the induction start $i=0$ we show $\mZpredCand_{0|k+1,N}\in\XX_0=\XX$
	using the row sum norm $\mNorm{\infty}{A_x}$ and the fact that
	$\mNorm{2}{a}\mNorm{2}{b}\leq \mNorm{1}{a}\mNorm{2}{b}$ for all $a,b\in\RR^{n_x}$ to get
	\begin{align*}
		A_x \mZpredCand_{0|k+1} & = A_x \mZpredOpt_{1|k,N} +  A_x(\mZpredCand_{0|k+1}-\mZpredOpt_{1|k,N})\\
			& \leq (1-\epsilon)\mOnes{n_x} + \mNorm{\infty}{A_x}\mNorm{2}{\mZpredCand_{0|k+1}-\mZpredOpt_{1|k,N}}\mOnes{n_x}.
	\end{align*}
	Since $V(\mZpredCand_{0|k+1},\mZpredOpt_{1|k,N}) \geq c_l ||\mZpredCand_{0|k+1}-\mZpredOpt_{1|k,N}||_2^2$
	we have with \eqref{eq:lemma_feasible_solution_proof_0}, \eqref{eq:lemma_feasible_solution_proof_1} that
	$||\mZpredCand_{0|k}-\mZpredOpt_{1|k,N}||_2\leq\sqrt{\frac{c_u}{c_l}}c_\gamma\gamma$ and therefore
		$A_x \mZpredCand_{0|k+1} \leq (1-\epsilon)\mOnes{n_x} + \mNorm{\infty}{A_x}\sqrt{\frac{c_u}{c_l}}c_\gamma\gamma\mOnes{n_x}$.
	Selecting $\gamma \leq \gamma_2 \mDef c_\gamma^{-1}\sqrt{\frac{c_l}{c_u}}\frac{\epsilon}{\mNorm{\infty}{A_x}}$ yields
	\begin{align}
		A_x \mZpredCand_{0|k+1} & \leq (1-\epsilon)\mOnes{n_x} + \epsilon \mOnes{n_x} \Rightarrow \mZpredCand_{0|k+1}\in\XX.
		\label{eq:lemma_feasible_solution_proof_2}
	\end{align}
	In order to show the induction step
	$
			\forall j\in\mIntInt{0}{i}:~ \mZpredCand_{j|k+1}\in\bar\XX_{j}	\Rightarrow	\mZpredCand_{i+1|k+1} \in\bar\XX_{i+1},
	$
	for all $i\in\mIntInt{0}{N-3}$ we use Assumption~\ref{ass:local_incremental_stabilizability}
	and derive
	\begin{align*}
		c_uc_\gamma^2\gamma^2
		 & \geq V(\mZpredCand_{0|k+1},\mZpredOpt_{1|k,N})              
		  \geq \rho^{-1} V(\mZpredCand_{1|k+1},\mZpredOpt_{2|k,N}) \geq ..    \\
		 & \geq \rho^{-i} V(\mZpredCand_{i|k+1},\mZpredOpt_{i+1|k,N}) \\
		 & \geq \rho^{-1-i} V(\mZpredCand_{i+1|k+1},\mZpredOpt_{i+2|k,N})
	\end{align*}
	and consequently
	\begin{align*}
		V(\mZpredCand_{i+1|k+1},\mZpredOpt_{i+2|k,N})\leq \rho^{i+1} c_u c_\gamma^2\gamma^2\leq \delta
		\text{ for all } i\in\mIntInt{0}{N-3}
	\end{align*}
	since $\rho \in (0,1)$. By Assumption \eqref{ass:local_incremental_stabilizability} we have
	$c_l ||\mZpredCand_{i+1|k+1}-\mZpredOpt_{i+2|k,N}||_2^2 \leq V(\mZpredCand_{i+1|k+1},\mZpredOpt_{i+2|k,N})$ and
	\begin{align}\label{eq:deviation_candidate_optimal}
		\mNorm{2}{\mZpredCand_{i+1|k+1}-\mZpredOpt_{i+2|k,N}}^2 & \leq \rho^{i+1} \frac{c_u}{c_l}c_\gamma^2\gamma^2.
	\end{align}
	Similar to~\eqref{eq:lemma_feasible_solution_proof_2} we can conclude with $\gamma\leq \gamma_2$
	that
	\begin{align*}
		A_x \mZpredCand_{i+1|k+1} & \leq A_x\mZpredOpt_{i+2|k} + \mNorm{\infty}{A_x} \sqrt{\rho^{i+1}\frac{c_u}{c_l}}c_\gamma\gamma \mOnes{n_x}            \\
								& \leq (1 - \epsilon_{i+2})\mOnes{n_x} + \sqrt{\rho^{i+1}}\epsilon\mOnes{n_x} \\
								& \leq (1 - \epsilon_{i+1})\mOnes{n_x} \Rightarrow \mZpredCand_{i+1|k+1}\in\bar\XX_{i+1}
	\end{align*}
	for all $i\in\mIntInt{0}{N-3}$, which proves constraint satisfaction of the candidate
	state sequence $\mZpredCand_{i|k+1}$ with respect to state constraints by induction.\\
	\emph{Step 2:} Regarding the terminal constraint~\eqref{eq:MPSF_opt_safety_constraint}, note that
	due to $\mZpredCand_{N-2}\in\bar\XX_{N-2}\subset\XX$ we can use
	Assumption~\ref{ass:local_incremental_stabilizability} similarly to before in order to obtain
	\begin{align*}
		\mNorm{2}{\mZpredCand_{N-1|k+1}-\mZpredOpt_{N|k,N}}^2 & \leq \rho^{N-1} \frac{c_u}{c_l}c_\gamma^2\gamma^2.
	\end{align*}
	Let $\gamma \leq \gamma_3 \mDef c_\gamma^{-1}\sqrt{\frac{c_l}{c_u}}\frac{\epsilon}{L_\mSafe}$,
	implying similarly
	\begin{align*}
		a_\mSafe (\mZpredCand_{N-1|k+1})
		 & \leq a_\mSafe(\mZpredOpt_{N|k}) + L_\mSafe \sqrt{\rho^{N-1}\frac{c_u}{c_l}}c_\gamma\gamma \mOnes{n_s} \\
		 &\leq (1 - \epsilon_{N-1})\mOnes{n_\mSafe},
	\end{align*}
	showing terminal constraint satisfaction of \eqref{eq:MPSF_opt_safety_constraint}.
	Next we consider input constraints. Let
	$\gamma \leq \gamma_4 \mDef c_\gamma^{-1}\sqrt{\frac{c_l}{c_u}}\frac{\epsilon}
		{\mNorm{\infty}{A_u} \pi_{\mathrm{max}}}$,
	yielding together with Assumption~\ref{ass:local_incremental_stabilizability}
	and~\eqref{eq:deviation_candidate_optimal}
	\begin{align*}
		\mNorm{2}{\mUpredCand_{i|k+1}-\mVpredOpt_{i+1|k,N}}^2
		                                                      & \leq \pi_{\mathrm{max}}^2\mNorm{2}{\mZpredCand_{i|k+1}-\mZpredOpt_{i+1|k,N}}^2 \\
															  & \leq \pi_{\mathrm{max}}^2\rho^{i} \frac{c_u}{c_l}c_\gamma^2\gamma^2,
	\end{align*}
	providing that
	\begin{align*}
		A_u \mUpredCand_{i|k+1}
		 & \leq A_u\mVpredOpt_{i+1|k} + \mNorm{\infty}{A_u} \pi_{\mathrm{max}}\sqrt{\rho^i\frac{c_u}{c_l}}c_\gamma\gamma\mOnes{n_u} \\
		 & \leq (1-\epsilon_i)\mOnes{n_u} \text{ for all } i\in\mIntInt{0}{N-2},
	\end{align*}
	showing input 
	constraint satisfaction \eqref{eq:MPSF_opt_input_constraints} of the candidate input sequence. For the
	uncertainty constraint \eqref{eq:MPSF_opt_uncertainty_constraints} we have by
	Lemma~\ref{lem:containment_following_confidence_set} that
	$
		 \mSetValuedConfidence{p_\mSafe}(\mZpredCand_{i|k+1}, \mUpredCand_{i|k+1})\subseteq
		 	\bar{\mCS}_{i+1}^{\gamma,+} (\mNorm{2}{\Delta z_i})
	$
	holds with $\bar{\mCS}_{i+1}^{\gamma,+}(\mNorm{2}{\Delta z_i})\mDef\{ e\in\RR^n | a_\mCS(e) \leq \gamma(1-\epsilon_{i+1})\mOnes{n_\mCS}
		+ L_{a_\mCS}L_{\mSetValuedConfidence{p_\mSafe}} \mNorm{2}{\Delta z_i}\mOnes{n_\mCS}\}$
	and
	\begin{align*}
		\Delta z_i &\mDef [{\mZpredCand_{i|k+1}}^\top, {\mUpredCand_{i|k+1}}^\top]^\top -[{\mZpredOpt_{i+1|k}}^\top, {\mVpredOpt_{i+1|k}}^\top]^\top,\\
		\mNorm{2}{\Delta  z_i}^2
		 & \leq \mNorm{2}{\mZpredCand_{i|k+1} - \mZpredOpt_{i+1|k}}^2 + \mNorm{2}{\mUpredCand_{i|k+1} - \mVpredOpt_{i+1|k}}^2 \\
		 & \leq (1 + \pi_{\mathrm{max}})^2 \rho^i\frac{c_u}{c_l}c_\gamma^2\gamma^2
	\end{align*}
	for all $i\in\mIntInt{0}{N-2}$. For a sufficiently small Lipschitz constant
	in Assumption~\ref{ass:set_valued_model_confidence_map}, such that $L_{\mSetValuedConfidence{p_\mSafe}} \leq c \epsilon$ with
	\begin{align}\label{eq:set_valued_confidence_map_lipschitz_bound}
		c \mDef \frac{1}{L_{a_\mCS}(1+\pi_{\mathrm{max}})}\sqrt{\frac{c_l}{c_u}}c_\gamma^{-1}
	\end{align}
	holds, we obtain
	\begin{align*}
		 & \mSetValuedConfidence{p_\mSafe}(\mZpredCand_{i|k+1}, \mUpredCand_{i|k+1})\subseteq\mCS_{i+1}^{\gamma,+}(\mNorm{2}{\Delta z}) \\
		 & ~\subseteq \{ e\in\RR^n | a_\mCS(e) \leq \gamma(1-\epsilon_{i+1})\mOnes{n_\mCS}                                                 \\
		 & \qquad\qquad~~~~+ L_{a_\mCS}L_{\mSetValuedConfidence{p_\mSafe}} (1 + \pi_{\mathrm{max}}) \sqrt{\rho^i\frac{c_u}{c_l}}c_\gamma\gamma \mOnes{n_\mCS}\}         \\
		 & ~\subseteq \{ e\in\RR^n | a_\mCS(e) \leq \gamma(1-\epsilon_{i+1})\mOnes{n_\mCS} + \sqrt{\rho^i}\epsilon \gamma\mOnes{n_\mCS}\}        \\
		 & ~\subseteq \{ e\in\RR^n | a_\mCS(e) \leq \gamma(1-\epsilon_{i}) \mOnes{n_\mCS}\}     \subseteq \bar \mCS^\gamma_i,
	\end{align*}
	which shows feasibility of \eqref{eq:MPSF_opt_uncertainty_constraints}.
	Therefore, if $L_{\mSetValuedConfidence{p_\mSafe}}$ satisfies
	$L_{\mSetValuedConfidence{p_\mSafe}} \leq c \epsilon$ for a selected
	$\epsilon$ with $c$ as defined in \eqref{eq:set_valued_confidence_map_lipschitz_bound},
	then there always exists a
	\begin{align}\label{eq:app_maximum_error_magnitude}
		\gamma \mDef \min\{\gamma_1,\gamma_2,\gamma_3,\gamma_4\}>0,
	\end{align}
	such that $\{\mUpredCand_{i|k+1}\}$ is a feasible solution to
	\eqref{eq:MPSF_opt} at time $k+1$ with planning horizon $N-1$,
	proving the desired statement.
\end{proof}
Note that the bounds $\gamma_i$ for $i=1,2,3,4$
unveil an inner relationship between the constraint tightening
fraction $\epsilon$ and the maximal allowable error set
that is defined through the scaling parameter $\gamma$.
Since the bound on $\gamma$ takes the form $\gamma \leq \min\{c_\gamma^{-1}\sqrt{\delta c_u^{-1}}, \min_{i\in\{2,3,4\}} c_i\epsilon\}$
it follows that a bound exists on the error magnitude
that can be tolerated depending on the size of
the region, for which local incremental stabilizability
according to Assumption~\ref{ass:local_incremental_stabilizability}
holds. Therefore, increasing the constraint tightening beyond
$\epsilon = c_\gamma^{-1}\sqrt{\delta c_u^{-1}}/\min_i c_i$ will not necessarily allow an
increase in the tolerable uncertainty magnitude.
We can now utilize Lemma~\ref{lem:feasible_solution} to
prove Theorem~\ref{thm:PSF}, i.e. that application of Algorithm 2
implies safe system operation according to~\eqref{eq:chance_constraints}.
\begin{proof}[Proof of Theorem~\ref{thm:PSF}]
	Select $\gamma$ according to Lemma~\ref{lem:feasible_solution}.
	By considering the different cases in Algorithm~2, we show safety
	$\forall k$ according to \eqref{eq:chance_constraints} by utilizing
	the relation
	\begin{align}
		& \Pr(\forall k: x(k)\in\XX,u(k)\in\UU)      \nonumber                          \\
   \geq & \Pr(\forall k:x(k)\in\XX,u(k)\in\UU,	
				   e(k,\theta_{\mathcal R})\in\mSetValuedConfidence{p_\mSafe}(k))  \nonumber			\\
   =    & \Pr(\forall k:x(k)\in\XX,u(k)\in\UU|	
				   \forall k: e(k,\theta_{\mathcal R})\in\mSetValuedConfidence{p_\mSafe}(k))     \nonumber\\
		   & \cdot \Pr(\forall k: e(k,\theta_{\mathcal R})\in\mSetValuedConfidence{p_\mSafe}(k)),\label{eq:thm_proof_2}
\end{align}
	where we use $\mSetValuedConfidence{p_\mSafe}(k)$
	instead of $\mSetValuedConfidence{p_\mSafe}(x(k),u(k))$.
	Since $\Pr(\forall k: e(k,\theta_{\mathcal R})\in\mSetValuedConfidence{p_\mSafe}(k))\geq p_\mSafe$
	by Assumption~\ref{ass:set_valued_model_confidence_map}, relation
	\eqref{eq:thm_proof_2} allows us to prove (3) by establishing
	\begin{align}\label{eq:thm:proof_3}
		\Pr(\forall k:x(k)\in\XX,u(k)\in\UU|	
			\forall k: e(k,\theta_{\mathcal R})\in\mSetValuedConfidence{p_\mSafe}(k))=1.
	\end{align}
	The proof therefore reduces to the deterministic case showing that
	$x(k)\in \XX$, $u(k) \in \UU$, given $e(k,\theta_{\mathcal R})\in\mSetValuedConfidence{p_\mSafe}(k))$,
	at any time step $k$, which implies directly \eqref{eq:thm:proof_3}
	for all times $k$ and therefore via \eqref{eq:thm_proof_2} chance constraint
	satisfaction according to \eqref{eq:chance_constraints}, i.e.,
	safe system operation with respect to~\eqref{eq:chance_constraints}.
	In order to show $x(k)\in\XX$ and $u(k)\in\UU$, note that if \eqref{eq:MPSF_opt}
	is feasible at time $k$ for any planning horizon $\tilde N>0$ it follows
	due to the state and input constraints \eqref{eq:MPSF_opt_state_constraints},
	\eqref{eq:MPSF_opt_input_constraints} that $x(k)\in\XX$ and $u(k)\in\UU$.
	This implies directly that $x(k)\in\XX$ and $u(k)\in\UU$ for any time step
	$k$ for which \eqref{eq:MPSF_opt} is feasible for horizon $N$, as well as for
	all $\tilde k \in \mIntInt{k+1}{N + k - 1}$, for which \eqref{eq:MPSF_opt} is
	infeasible for horizon $N$, since feasibility of \eqref{eq:MPSF_opt} with
	horizon $N-(\tilde k - k)$ is obtained from iteratively
	applying Lemma~\ref{lem:feasible_solution} due to the condition that
	$\forall k: e(k,\theta_{\mathcal R})\in\mSetValuedConfidence{p_\mSafe}(k)$.
	For all $\tilde k\geq N+ k$, it follows from containment in the
	terminal safe set via \eqref{eq:MPSF_opt_safety_constraint} and
	 Assumption~\ref{ass:terminal_safe_set} that
	 $x(\tilde k)\in\XX$ and $u(\tilde k)=\pi_\mSafe^t(x(k))\in\UU$.
	This shows \eqref{eq:thm:proof_3} and therefore via \eqref{eq:thm_proof_2} that
	$\Pr(\forall k: x(k)\in\XX,u(k)\in\UU)\geq p_\mSafe$, completing the proof.
\end{proof}
\subsection{Offline design verification}\label{app:design_certification}
	The intuitive interpretation of the few parameters that
	need to be chosen allows to propose a selection and to
	certify it using sampling. Related statistical verification
	methods have been presented, e.g., in~\cite{Hertneck2018,karg2019probabilistic}.
	More precisely, to verify the underlying assumptions of Theorem~\ref{thm:PSF}
	and to ensure that the admissible error, i.e. $\gamma$
	in~\eqref{eq:admissible_error}, is selected sufficiently small, we
	make use of a statistical offline verification procedure
	for finite horizon tasks, i.e. $\bar N < \infty$ in~\eqref{eq:general_objective},
	for parameter design given a particular realization of a PSF $\pi_\mSafe$
	together with a learning policy $\pi_\LL$.
	We sample and simulate a sufficiently large
	number $N_\mathrm{total}\in\mathbb N^+$ of system parameter realizations
	as well as initial conditions according to their distributions
	$p(\theta)$ and $p(x_{\mathrm{init}})$ and provide
	a statistical bound on the number of successful simulations
	that ensure safety according to~\eqref{eq:chance_constraints}.
	Let $ \theta^{(j)}\sim p(\theta)$ and
	$ x_{\mathrm{init}}^{(j)}\sim p(x_\mathrm{init})$ for $j=1,2,..,N_\mathrm{total}$
	be i.i.d. samples and define an indicator function
	for safe execution as
	\begin{align*}
		\mIndi( \theta^{(j)},  x_{\mathrm{init}}^{(j)})
		=
		\begin{cases}
			1,\forall k=1,..,\bar N: & x_{\mathrm{init}}^{(j)}(k)\in\XX~\land \\
				& \pi_\LL(k,x_{\mathrm{init}}^{(j)}(k))\in\UU\\
			0,&\text{otherwise}
		\end{cases}
	\end{align*}
	with sampled system dynamics of the form
	$x_{\mathrm{init}}^{(j)}(k+1)=f(x_{\mathrm{init}}^{(j)}(k),
		\pi_\LL(k,x_{\mathrm{init}}^{(j)}(k));
		\theta^{(j)})$
	and initial condition $x_{\mathrm{init}}^{(j)}(0)=x_\mathrm{init}^{(j)}$.
	The probability of satisfying constraints $\XX$ and $\UU$ can therefore be expressed as
	$p_\mSafe^\infty \mDef \lim_{N_{\mathrm{total}}\rightarrow\infty}\bar p_\mSafe(N_\mathrm{total})$
	with
	\begin{align*}
		\bar p_\mSafe(N_\mathrm{total}) \mDef N_{\mathrm{total}}^{-1} \sum_{j=0}^{N_{\mathrm{total}}} \mIndi(\theta^{(j)}, x_{\mathrm{init}}^{(j)})
	\end{align*}
	via the law of large numbers. To determine a sufficiently large but finite
	number $N_{\mathrm{total}}$ to ensure $p_\mSafe^\infty \geq p_\mSafe$ we
	use Hoeffding's inequality \citep[Section 4.1]{Luxburg2011} as similarly proposed
	in~\citet{Hertneck2018}:
	$
		\Pr(|p_\mSafe^\infty - \bar p_\mSafe(N_\mathrm{total})| \geq \Delta_H^2)
		\leq 2 \exp (-2 N_{\mathrm{total}}\Delta_H)
	$
	for an error margin $\Delta_H>0$.
	By defining $\delta_H\mDef 2 \exp (-2 N_{\mathrm{total}}\Delta_H)$ it follows
	that $p_\mSafe^\infty \geq p_\mSafe(N_\mathrm{total}) + \Delta_H$ holds
	at confidence level $1-\delta_H$.
	This allows us to state a formal lower bound on the total
	number of offline simulations as follows:
	\begin{proposition}
		Consider a specific PSF problem parametrization,
		learning policy $\pi_\LL$, and
		confidence level $1-\delta_H$. If 
		$\bar p_\mSafe(N_\mathrm{total}) > p_\mSafe$ holds and
		\begin{align}\label{eq:hoeffdings_design_number}
			N_\mathrm{total} \geq  - \frac{\log (0.5\delta_H)}{2(\bar p_\mSafe(N_\mathrm{total})-p_\mSafe)^2},
		\end{align}
		then the chance constraints~\eqref{eq:chance_constraints} are fulfilled under
		application of $u(k)=\pi_\mSafe(k, x(k), \pi_\LL(k,x(k)))$
		to the real system~\eqref{eq:general_nonlinear_system}.
	\end{proposition}
	For example, if we have 
	$p_\mSafe = 0.95$ and $\bar p_\mSafe(N_\mathrm{total}) = 0.99$
	 for a specific choice of design parameters, then
	$N_{\mathrm{total}}$ needs to be greater than $1656$ to provide
	a sound PSF parametrization with $99 \%$ confidence.
	Note that~\eqref{eq:hoeffdings_design_number}
	is independent of the complexity of the
	system and that offline simulations can be executed in parallel.
	Naturally, small required values of $|p_\mSafe  - \bar p_\mSafe(N_\mathrm{total})|$ can lead
	to a rapid increase of $N_{\mathrm{total}}$. In case of infeasibility
	one needs to re-adjust the design parameters as described in Section~\ref{subsec:design_parameters}
	or more data needs to be collected.

\subsection{Sufficient condition for Assumption~\ref{ass:local_incremental_stabilizability}}
	Different from \citet{Koehler2018a,Koehler2018b} we require the first condition in
	Assumption~\ref{ass:local_incremental_stabilizability}
	to also hold for the case that $V(x,z,v)>\delta$, which is used in the proof of Lemma~\ref{lem:feasible_solution}.
	Nevertheless, \citet[Prop. 1]{Koehler2018a} still implies that the following verifiable
	assumption is sufficient for Assumption~\ref{ass:local_incremental_stabilizability}
	and relates it to local stabilizability of system \eqref{eq:general_nonlinear_system}.
	\begin{assumption}\label{ass:sufficient_condition_locally_incremental_stabilizable}
		Let $r \mDef (\mu,v) \in \XX \times \UU$
		and define the linearization $A_r \mDef \tfrac{\partial f}{\partial x}(r,\bar \theta)$,
		$B_r \mDef \tfrac{\partial f}{\partial u}(r,\bar \theta)$.
		For any $r\in\XX\times\UU$, the pair $(A_r,B_r)$
		is stabilizable, i.e. there exist
		$K_r\in\RR^{m\times n}$, $P_r,Q_r\in\RR^{n \times n}$ positive definite
		and continuous in $r$, such that
		\begin{align}\label{eq:ricatti_decrease}
			P_r - (A_r + B_r K_r)^\top P_r (A_r + B_r K_r) = Q_r.
		\end{align}
		Furthermore, there exists a constant
		$c\in\RR_{>0}$, such that for any $r^+=(\mu^+,v^+)\in\XX\times\UU$ with $\mu^+ = f(\mu,v,\bar \theta)$, the
		corresponding matrix $P_{r^+}$ satisfies:
		\begin{align}\label{eq:bound_V_rate}
			\lambda_{\mathrm{max}}(P_r^{-1}P_{r^+})Q_r \geq (\lambda_{\mathrm{max}}(P_r^{-1}P_{r^+})-1)P_r + cI_n.
		\end{align}
	\end{assumption}
	Given Assumption~\ref{ass:sufficient_condition_locally_incremental_stabilizable},
	we can choose $V(x,\mu,v)=(x-\mu)^\top P_r (x-\mu)$ in
	Assumption~\ref{ass:local_incremental_stabilizability}, with \eqref{eq:bound_V_rate}
	bounding the rate at which $V(x,\mu,v)$ can possibly change in any time step
	when applying $u=\pi(x,\mu,v)=v+K_r(x-\mu)$.
	Note that, according to \citet[Prop. 1]{Koehler2018a},
	Assumption~\ref{ass:sufficient_condition_locally_incremental_stabilizable}
	is only sufficient for Assumption \ref{ass:local_incremental_stabilizability},
	and therefore stabilizability of the linearization along any reference
	might not be needed in practice. However, it allows us to
	analytically express conditions that can be verified in principled steps.
	While \eqref{eq:ricatti_decrease} is always satisfied as long as the
	system is locally stabilizable, the second equation accounts for linearization errors
	as the system evolves and can therefore be seen as an intrinsic consequence
	of a local linearization-based analysis.
	As a practical consequence, the condition might be too conservative for
	nonlinear systems with quickly changing linearizations along
	references, i.e. large $||(A_{r+1},B_{r+1})-(A_r,B_r)||$. See, e.g., \citet{Koehler2018a,Koehler2018b}
	for concrete examples that demonstrate how 
	Assumption~\ref{ass:sufficient_condition_locally_incremental_stabilizable}
	can be verified for different nonlinear systems.
\end{document}